\documentclass[twocolumn]{IEEEtran}
\usepackage{blindtext}
\usepackage[nolist,nohyperlinks]{acronym}
\usepackage{hyperref}
\usepackage{svg}
\usepackage{amsthm}

\newtheorem{proposition}{Proposition}

\usepackage{amsmath,epsfig}
\usepackage[T1]{fontenc}
\usepackage[utf8]{inputenc} 
\usepackage{psfrag}
\usepackage{array}
\usepackage{cite}
\usepackage{amsfonts}

\usepackage{amssymb}
\usepackage{color}
\usepackage[dvipsnames]{xcolor}
\usepackage[table]{xcolor}
\usepackage{rotating}
\usepackage{subfigure}
\usepackage{graphicx}
\usepackage{float}
\usepackage{circuitikz}
\usepackage{tikz}
\usepackage{amsbsy} 
\usepackage{tcolorbox}
\usepackage{multirow}
\renewcommand{\arraystretch}{1.6}
\usepackage{acronym}

\usepackage{algorithm}
\usepackage{algorithmic}

\usepackage{amsthm}

\def\blfootnote{\xdef\@thefnmark{}\@footnotetext}
\makeatother

\newtheorem{remark}{Remark}

\makeatletter
\def\blfootnote{\xdef\@thefnmark{}\@footnotetext}
\makeatother

\input{acronym.def}

\title{\huge{Hybrid Multiport Receivers for Slow \\Fluid Antenna Multiple Access}}

\begin{document}
\begin{acronym}
\acro{PLS}{physical layer security}
\acro{OP}{outage probability}
\acro{SNR}{signal-to-noise ratio}
\acrodefplural{SNR}[SNRs]{signal-to-noise ratios}
\acro{AWGN}{additive white Gaussian noise}
\acro{CSI}{channel state information}
\acro{PDF}{probability density function}
\acrodefplural{PDF}[PDFs]{probability density functions}
\acro{CDF}{cumulative distribution function}
\acro{PDF}{probability density function}
\acrodefplural{CDF}[CDFs]{cumulative distribution functions}
\acro{LRS}{large reflecting surface}
\acro{BS}{base station}
\acro{RF}{radio frequency}
\acro{LOS}{line-of-sight}
\acro{MC}{Monte Carlo}
\acro{RV}{random variable}
\acro{CDF}{cumulative density function}
\acro{SINR}{signal-to-interference-plus-noise ratio}
\acro{SNR}{signal-to-noise ratio}
\acro{FAS}{fluid antenna system}
\acro{MIMO}{multipe-input multiple-output}
\acro{BS}{base station}
\acro{CSI}{channel state information}
\acro{LoS}{line-of-sight}
\acro{RF}{radio frequency}
\acro{SE}{spectral efficiency}
\acro{OP}{outage probability}
\end{acronym}
\author{Jos\'e P. Gonz\'alez-Coma, Jos\'e D.~Vega-S\'anchez,~\IEEEmembership{Senior Member,~IEEE} and F. Javier L\'opez-Mart\'inez,~\IEEEmembership{Senior Member,~IEEE}}


\maketitle

\begin{abstract}
We propose a novel receiver architecture that preserves the performance benefits of multiport selection in fluid-antenna systems while requiring only a very small number of radio-frequency (RF) chains. The resulting fluid-antenna hybrid multiport (FAHM) receiver effectively decouples port selection from signal combining by integrating a low-complexity analog combining network similar to those used in conventional hybrid multiantenna designs. {We develop an effective-port criterion to determine the selected-port dimension before hybrid combining, thereby limiting the performance loss associated with port selection,} and then design the hybrid combiner for a given RF-chain budget. The FAHM architecture is evaluated in a multiuser set-up operating under slow fluid-antenna multiple access (FAMA). In this scenario, {a FAHM implementation with only 2 RF chains showcases a performance comparable to a fully-digital conventional multiport scheme with a much larger number of RF chains.} {Additionally, the proposed Rank-1 implementation attains a reduction in average GEPort execution time by about two orders of magnitude.}
\end{abstract}


\blfootnote{\noindent Manuscript received XX XX, XXXX. The review of this paper was coordinated by XXXX. This work is supported by grant PID2023-149975OB-I00 (COSTUME) funded by MICIU/AEI/10.13039/501100011033 and FEDER/UE, and  by grant PICUD-2025-02 (COMTEUM) funded by the Defense University Center at the Spanish Naval Academy. Also, this work is supported by the Universidad San Francisco de Quito through the Poli-Grants Program under Grant 41989. This work has been submitted to the IEEE for publication, after which this version can no longer be accessible. 
}

\blfootnote{\noindent J.P. Gonz\'alez-Coma is with the Defense University Center at the Spanish Naval Academy, 36920 Marín, Spain. Contact email: $\rm jose.gcoma@cud.uvigo.es$.}
\blfootnote{\noindent 
J.~D.~Vega-S\'anchez is with Colegio de Ciencias e Ingenier\'ias  ``El Polit\'ecnico", Universidad San Francisco de Quito (USFQ), Diego de Robles S/N, Quito (Ecuador) 170157.  Contact email: $\rm dvega@usfq.edu.ec$.
}
\blfootnote{\noindent F.J. L{\'o}pez-Mart{\'i}nez is with Dept. Signal Theory, Networking and Communications, Research Centre for Information and Communication Technologies (CITIC-UGR), University of Granada, 18071, Granada, (Spain). Contact e-mail: $\rm fjlm@ugr.es$.
} 

\blfootnote{The MATLAB code used for simulations and figures in this paper will be available on Github upon
publication of the manuscript.}

\blfootnote{Corresponding author: J.P. Gonz\'alez-Coma.}

\vspace{0mm}
\section{Introduction}
\IEEEPARstart{T}{o} support the extreme connectivity envisioned for beyond-fifth generation (B5G) and  sixth-generation (6G) networks, the physical layer must accommodate a large number of users sharing the same time-frequency resources under latency and interference constraints~\cite{Clerckx2024MA6G}. In this context, recent research has revisited spectrum-sharing strategies for massive access, ranging from advanced multiple-access frameworks to very large-scale spatial multiplexing. However, these solutions rely heavily on accurate channel state information (CSI) acquisition and increasingly complex joint transceiver-receiver processing, limiting scalability in dense deployments~\cite{Wang2024XLMIMO}.

Fluid antenna systems (FASs) have emerged as a promising alternative dealing with the radiating aperture as a reconfigurable physical-layer resource rather than as a fixed antenna structure. The original vision of FAS was introduced in~\cite{Wong2021}, while more recent tutorial and survey works have broadened its theoretical, hardware, and networking foundations~\cite{BruceLee2020,Lu2025FluidAntennas,Hong2026Survey}. By exploiting the position reconfigurability of FAS, \ac{FAMA} was proposed as a transmitter-\ac{CSI}-free multiple-access strategy in which each user searches for favorable spatial positions where interference is naturally suppressed~\cite{FAMA}. Later, in the state-of-the-art, several research directions have emerged. For instance, fast-\ac{FAMA} was introduced in~\cite{FastFAMA}, where each user equipment (UE) employs a single fluid antenna and dynamically switches its port on a symbol-by-symbol basis so as to strengthen the desired signal while suppressing interference. Although this strategy offers strong multiplexing capabilities, its practical implementation is challenging because it requires extremely fast and reliable port reconfiguration. To alleviate this limitation, the more practical slow-\ac{FAMA} scheme was proposed in~\cite{SlowFAMA}, where the selected port is updated only when the fading channel changes, so that the chosen port remains fixed over a coherence interval. Building on this foundation, subsequent works refined the theoretical understanding of \ac{FAMA} through outage analysis for the two-user case~\cite{Xu2024Outage}, explored data-driven receiver designs via deep learning for slow-\ac{FAMA}~\cite{waqar2023dlslow}, and extended the framework to coded transmissions over both block-fading and fast-fading channels~\cite{Hong2025BlockCodedFAMA,Hong2025FastCodedFAMA}. More recently, enhanced receiver strategies have also been investigated to further boost the interference-management capability of \ac{FAMA} systems~\cite{Waqar2025TurboFAMA}.

{From a hardware perspective, FASs should not be interpreted only as mechanically port-selectable antennas. Liquid-metal and fluidic implementations have been investigated as a way to physically reconfigure the radiating structure and its operating characteristics~\cite{Fang2025LiquidMetal}. Pixel-based electronically reconfigurable antennas have also been proposed to realize fast radiation-state switching for FAS operation~\cite{Zhang2024}. In parallel, metasurface and meta-fluid implementations have recently been explored as programmable electromagnetic platforms for controlling the effective radiation state of the antenna~\cite{Liu2025}. These hardware-oriented developments complement broader FAS tutorial and survey discussions, which emphasize that fluid antennas can be implemented through different mechanical, material, and electronically reconfigurable mechanisms~\cite{Lu2025FluidAntennas}. In this work, we focus on the port-selectable slow-FAMA abstraction, where a finite set of candidate ports represents the available spatial states of the FA receiver, while recognizing that this abstraction is one implementation-oriented model within the broader FAS hardware landscape.}

On the other hand, to further improve the multiplexing capability of slow-\ac{FAMA}, recent works have moved from single-port reception to multiport architectures. In particular, the compact ultra-massive antenna array (CUMA) showed that activating multiple ports through analog aggregation can significantly improve open-loop massive connectivity~\cite{CUMA}, and its extension to up to four \ac{RF} chains demonstrated that additional gains can be harvested under practical handset RF-chain constraints~\cite{CUMA2}. {More recently, the phased ultra massive array (PUMA) framework further strengthened the theoretical basis of analog multiport reception by integrating phase shifters into the analog combining stage of CUMA, thereby enabling more structured phased aggregation with a limited number of RF chains~\cite{PUMA26}.} {However, the achievable performance of CUMA is highly dependent on a heuristic selection of the desired set of ports, and its performance limits are not well understood. In \cite{Hong25a}, a multiport implementation of FAMA based on a sequential \textit{select-then-combine} operation for the port selection and combination stages was proposed, implementing an interference-rejection digital combining strategy. Very recently,} multiport slow-FAMA receivers based on \textit{joint port selection and combining} were investigated in~\cite{GoLo26}, where substantial \ac{SINR} gains were reported even with a small number of \ac{RF} chains. {Following up this later work, other alternatives for} multi-port selection algorithms have been proposed to improve the performance-complexity trade-off in slow-\ac{FAMA} systems with multiple active ports: for instance,~\cite{hong2025multiport} analyzed exhaustive-search port selection (EPS), incremental port selection (IPS), and decremental port selection (DPS) for receivers employing interference rejection combining (IRC). Here, EPS serves as a performance benchmark by exhaustively evaluating the candidate port subsets, whereas IPS and DPS provide lower-complexity alternatives that retain most of the achievable gain in practical fluid-antenna deployments. {Finally}, the multi-port selection problem has also been revisited from both algorithmic and learning-based perspectives in~\cite{perezadan2026greedy}, where a greedy forward-selection method with swap refinement and a Transformer-based neural-network solution were proposed. These results showed that learning-aided strategies can approach the spectral-efficiency performance of strong greedy baselines while reducing inference complexity.

{Despite this progress, FA multiport schemes face a number of challenges associated with complexity and performance trade-offs. }Existing multiport slow-\ac{FAMA} receivers either tightly couple the number of selected ports to the number of available \ac{RF} chains in digital combining \cite{GoLo26,hong2025multiport}, or rely on {rigid and heuristic analog aggregation mechanisms with reduced numbers of RF chains} \cite{CUMA}. {To the best of our knowledge, the problem of decoupling the selected-port dimension from the hardware constraints that limit the number of available RF chains at user terminals remains unsolved. {This decoupling is important because the number of selected ports and the number of RF chains play different roles in a practical FA receiver. A larger selected-port set allows the receiver to retain more port-domain spatial information, which is beneficial for exploiting diversity and mitigating multiuser interference. By contrast, the number of RF chains is a hardware-limited quantity that directly affects cost, power consumption, and implementation complexity at the user terminal. Therefore, imposing $L=P$ forces}{\footnote{{Here, $P$ denotes the number of selected FA ports and $L$ denotes the number of available RF chains.}}}{ the receiver either to discard potentially useful FA ports when only a few RF chains are available, or to increase the RF-chain count in order to exploit a larger port set. The proposed FAHM architecture avoids this limitation by selecting $P$ ports according to the available spatial information while using only $L\ll P$ RF chains after analog combining.} In this work, we develop a framework that integrates the selection of multiple analog ports with a given budget of RF chains through an analog combining network akin to those classically used in hybrid \ac{MIMO} schemes~\cite{ZhMoKu05}. Specifically, we propose a fluid-antenna hybrid multiport (FAHM) receiver for slow-\ac{FAMA}. The FAHM architecture preserves the performance benefits of multiport selection while requiring only a very small number of \ac{RF} chains -- typically two. The selection of a subset of ports is formally integrated with the signal combiner architecture through a low-complexity analog/digital structure, effectively decoupling the dimensions of multiport selection and signal combining.}

{The analog combining network itself is 
structurally related to conventional hybrid MIMO. The distinction lies in the port-domain design introduced by the reconfigurable FA aperture. Whereas hybrid MIMO maps a fixed antenna array to fewer RF chains, FAHM operates across $N\geq P\geq L$: it determines a channel-dependent dimension $P$, selects those ports from $N$ candidate spatial states, and maps the retained subspace to $L$ RF chains. The proposed design therefore separates \emph{how many} ports to preserve, \emph{which} ports to retain, and \emph{how} to realize their hybrid combination.} {A second distinction is that, for the single-stream slow-FAMA receiver, the RF/baseband codesign in~\cite{ZhMoKu05} realizes the fully digital combining vector over any fixed selected-port subspace exactly with $L=2$, irrespective of $P$. We use this known decomposition only as an enabling implementation result. The proposed advances are the effective-port criterion for determining $P$, its integration with interference-aware GEPort selection, the Rank-1 implementation that accelerates GEPort without changing its decisions, and the unifying formulation that recovers CUMA as a structured instance of FAHM.}

The main contributions of this paper are summarized as:
\begin{itemize}
\item {We propose a FAHM receiver for slow-\ac{FAMA} that exactly realizes the fully digital combining vector over $P$ selected ports using only $L=2$ RF chains, irrespective of $P$, thereby preserving multiport reception with $L\ll P$.}
\item {We introduce an effective-port criterion, denoted as $P_{\mathrm{eff}}$, which provides a principled design rule to determine the minimum selected-port dimension required before hybrid combining to achieve a target performance under a fixed RF chain budget.}
\item We develop {a novel} implementation of the {state-of-the-art} Generalized Eigenvector Port (GEPort)~\cite{GoLo26} selection rule that preserves {key benefits of its original formulation} while {significantly} lowering the per-iteration computational cost.
\item We {formally} show that CUMA can be interpreted as a particular structured instance of the proposed FAHM framework. {This gives theoretical support to its inherently heuristic design, and provides a grounded theoretical framework that enables fair performance comparisons.}
\item  {We evaluate the performance of the proposed FAHM receiver under two different combining strategies: FAHM-based Digital Combining (DC), denoted as FAHM-DC, and FAHM-GEPort} in terms of \ac{SE} and \ac{OP}, showing consistent gains over conventional slow-\ac{FAMA} and CUMA benchmarks.
\end{itemize}

The remainder of the paper is organized as follows: Section \ref{sm} presents the system model and the proposed hybrid multiport receiver architecture for slow-\ac{FAMA}. Section \ref{RD} develops the receiver design, including the reduced-complexity GEPort implementation, the effective-port criterion, and the interpretation of CUMA as a special case of the proposed framework. Section \ref{SecIV} introduces the post-combining \ac{SINR} together with the key performance metrics, and presents the relevant numerical results.
 Section \ref{conclusiones} concludes the paper.

\textit{{Notation:}} Scalars are denoted by lowercase letters; vectors and matrices are denoted by boldface lowercase and uppercase letters, respectively. $\mathbf{I}_M$ denotes the identity matrix of dimension $M$. Transpose and Hermitian transpose are denoted by $(\cdot)^T$ and $(\cdot)^H$, respectively. Calligraphic letters denote sets, $|\cdot|$ denotes either set cardinality or complex modulus depending on the argument, and $\mathbb{E}\{\cdot\}$ denotes expectation. {$\|\cdot\|_2$ denotes the Euclidean norm, $\mathcal{N}_{\mathbb{C}}(\mu,\sigma^2)$ denotes a circularly symmetric complex Gaussian distribution with mean $\mu$ and variance $\sigma^2$, $(\cdot)^\dagger$ denotes the Moore--Penrose pseudoinverse, and $\Re\{\cdot\}$ and $\Im\{\cdot\}$ denote the real and imaginary
parts, respectively.}

\vspace{0mm}
\section{System Model}
\label{sm}
Consider {the \ac{DL} of a} \ac{MIMO} setup where the \ac{BS} is equipped with $M$ antennas and simultaneously serves $U$ users. Each user employs a hybrid FA array, as depicted in Fig.~\ref{fig:sysmodel}{\footnote{{The architecture in Fig.~\ref{fig:sysmodel} is conceptually related to the recently proposed tri-hybrid \ac{MIMO} paradigm, where signal processing is distributed across digital/baseband, analog/RF, and reconfigurable-antenna domains~\cite{Heath2026TriHybrid,Castellanos2026TriHybrid}. In the proposed FAHM receiver, the port-selection matrix $\mathbf{S}_u$ represents the antenna-domain reconfiguration stage, whereas $\mathbf{F}_u$ and $\mathbf{w}_u$ correspond to the analog and digital combining stages, respectively. Unlike general tri-hybrid \ac{MIMO} frameworks, FAHM is tailored to receiver-side slow-\ac{FAMA} operation with open-loop transmission and per-user port selection.}}}, with $N$ ports, from which $P$ ports can be {simultaneously} activated according to {the available \ac{CSI}}. These $P$ ports are connected through an analog network to $L$ \acf{RF} chains. {To implement an open-loop slow-\ac{FAMA}} scheme~\cite{SlowFAMA}, we {assume} the number of \ac{BS} antennas $M$ to be equal to the number of users $U$, i.e., $M=U$.

\begin{figure*}[!t]
\includegraphics[width=0.95\linewidth]{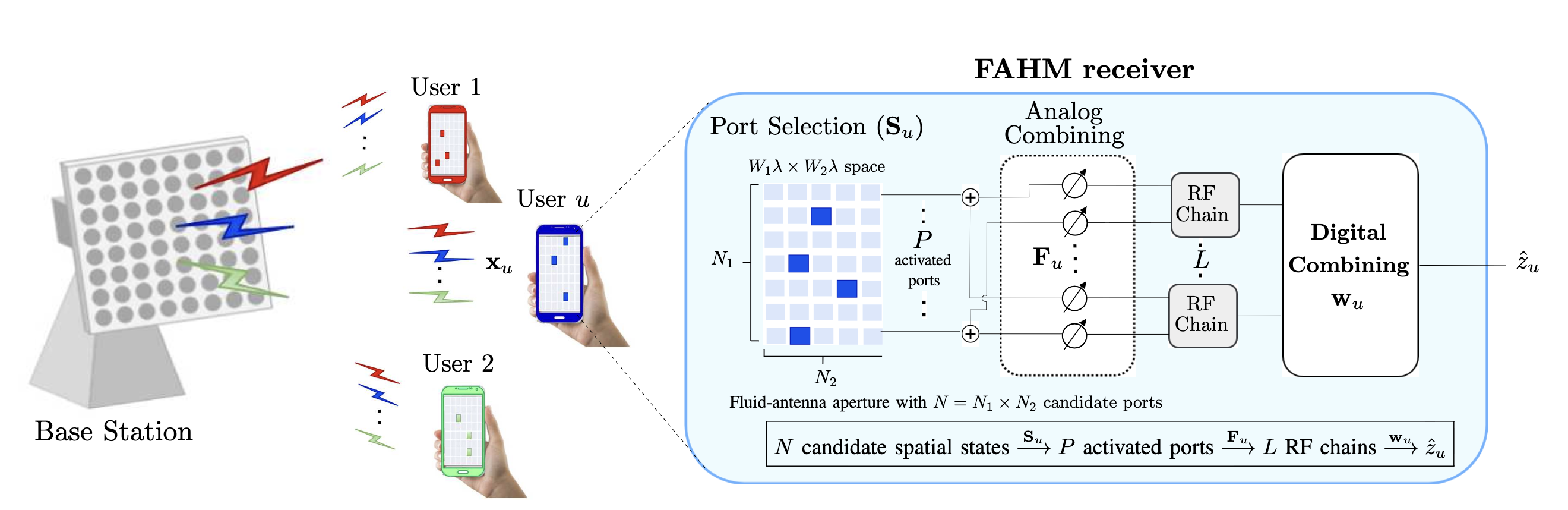}
\vspace{0mm}\caption{{System model of the proposed FAHM receiver.}}
\label{fig:sysmodel}
\end{figure*}
We denote by $z_u \in \mathbb{C}$ the data symbol intended for user $u$, with $\mathbb{E}\{|z_u|^2\}=\sigma_S^2$. Accordingly, the signal received at the $u$-th user can be expressed as
\begin{equation}
    \mathbf{x}_u = \mathbf{H}_u \mathbf{p}_u z_u + \sum_{j \neq u} \mathbf{H}_u \mathbf{p}_j z_j + \mathbf{n}_u,
    \label{eq1}
\end{equation}
where $\mathbf{H}_u \in \mathbb{C}^{N \times M}$ denotes the channel matrix between the \ac{BS} and user $u$, while $\mathbf{x}_u \in \mathbb{C}^{N \times 1}$ represents the corresponding received signal vector at the FA user. Although this vector representation may suggest a one-dimensional (1D) implementation of the FA array, the two-dimensional (2D) {is also captured through} an appropriate remapping of the 2D FA indices. Due to the low-complexity nature of slow-\ac{FAMA}, the transmitter operates without requiring \ac{CSI}~\cite{SlowFAMA}. Consequently, the precoding vector at the \ac{BS} is simply chosen as $\mathbf{p}_u=\mathbf{e}_u$, where $\mathbf{e}_u$ denotes the $M$-dimensional canonical vector whose entries are all zero except for the $u$-th one. Throughout this work, such a precoding vector is assumed to be known and fixed, {in coherence with the open-loop nature of slow-FAMA}. Finally, $\mathbf{n}_u \in \mathbb{C}^{N \times 1}$ denotes the additive white Gaussian noise vector, whose entries have variance $\sigma^2$. To characterize the received signal in~\eqref{eq1}, we adopt a finite-scatterer geometric channel model. In particular, the $m$-th column of the channel matrix $\mathbf{H}_u \in \mathbb{C}^{N \times M}$ is modeled as the superposition of one \ac{LoS} component and {multiple} scattered components, i.e. \cite{Rappaport2015Wideband}, 
\begin{align}
[\mathbf{H}_u]_m
= &~
\sqrt{\frac{K}{K+1}}\,e^{j\delta_{u,m}}
\mathbf{a}\!\left(\theta_{0}^{(u,m)},\phi_{0}^{(u,m)}\right)
\nonumber \\ &
+
\frac{1}{\sqrt{N_p}}\sqrt{\frac{1}{K+1}}
\sum_{\ell=1}^{N_p}
\kappa_{\ell}^{(u,m)}
\mathbf{a}\!\left(\theta_{\ell}^{(u,m)},\phi_{\ell}^{(u,m)}\right),
\label{eq:Hkm}
\end{align}
{where $K$ can be interpreted as a Rice-like factor, $\delta_{u,m}$ denotes the phase of the \ac{LoS} component, $N_p$ represents the number of scattered components, and $\kappa_{\ell}^{(u,m)}$ is the random complex coefficient of the $\ell$-th scattered path. Moreover, $\left(\theta_{0}^{(u,m)},\phi_{0}^{(u,m)}\right)$ and $\left(\theta_{\ell}^{(u,m)},\phi_{\ell}^{(u,m)}\right)$ denote the angular parameters associated with the \ac{LoS} path and the $\ell$-th scattered path, respectively. Notice that the \ac{LoS} term in \eqref{eq:Hkm} includes the unit-modulus phase factor $e^{j\delta_{u,m}}$, whereas the port-dependent geometric phase differences are captured by the steering vector $\mathbf{a}(\theta,\phi)$. For notational simplicity, the superscripts $(u,m)$ are omitted in the following geometric description. We adopt a Cartesian coordinate system in which $\theta_{\ell}\in[0,\pi]$ denotes the polar angle measured from the positive $z$-axis, while $\phi_{\ell}\in[0,2\pi)$ denotes the azimuth angle measured from the positive $x$-axis on the $xy$-plane. Accordingly, the unit vector pointing along the propagation direction of the $\ell$-th incident plane wave is given by}
\begin{equation}
{
\mathbf{u}_{\ell}
=
\left[
\sin\theta_{\ell}\cos\phi_{\ell},\,
\sin\theta_{\ell}\sin\phi_{\ell},\,
\cos\theta_{\ell}
\right]^T.}
\end{equation}
{Let $\mathbf{r}_t\in\mathbb{R}^3$ denote the position of the $t$-th FA port in the adopted user-side coordinate system and let $\mathbf{r}_1$ denote the reference port. Without loss of generality, the reference port is chosen as the origin, i.e., $\mathbf{r}_1=\mathbf{0}$. The relative propagation distance difference of the $\ell$-th path between the reference port and the $t$-th port is then defined as}
\begin{equation}
{
d_{\ell}(t)
=
\left(\mathbf{r}_t-\mathbf{r}_1\right)^T\mathbf{u}_{\ell}.}
\end{equation}
{Furthermore, $\mathbf{a}(\theta,\phi)$ stands for the array steering vector, which is given by}
\begin{equation}
\mathbf{a}(\theta_{\ell},\phi_{\ell})
=
\left[
1,\,
e^{-j\frac{2\pi}{\lambda}d_{\ell}(2)},\,
\ldots,\,
e^{-j\frac{2\pi}{\lambda}d_{\ell}(N)}
\right]^T,
\end{equation}
where $\lambda$ is the carrier wavelength. {For the 1D FA configuration, the $N$ ports are uniformly distributed along the $x$-axis over an aperture of length $W\lambda$, such that}
\begin{equation}
{
\mathbf{r}_t
=
\left[
\frac{(t-1)W\lambda}{N-1},\,0,\,0
\right]^T.}
\end{equation}
{Therefore, $d_{\ell}(t)$ can be expressed as}
\begin{equation}
d_{\ell}(t)
=
\frac{(t-1)W\lambda}{N-1}\sin\theta_{\ell}\cos\phi_{\ell}, \quad t=1,\dots,N
\end{equation}
where $W$ is the normalized length of the array. {For the 2D FA configuration, the ports are uniformly distributed over the $xz$-plane. In particular, the port indexed by $(t_1,t_2)$ is located at}
\begin{equation}
{
\mathbf{r}_{t_1,t_2}
=
\left[
\frac{(t_1-1)W_1\lambda}{N_1-1},\,0,\,
\frac{(t_2-1)W_2\lambda}{N_2-1}
\right]^T.}
\end{equation}
{Hence, the propagation difference is given by}
\begin{equation}
d_{\ell}(t)
=
\frac{(t_1-1)W_1\lambda}{N_1-1}\sin\theta_{\ell}\cos\phi_{\ell}
+
\frac{(t_2-1)W_2\lambda}{N_2-1}\cos\theta_{\ell},
\end{equation}
where $W_1$ and $W_2$ denote the {wavelength-}normalized horizontal and vertical dimensions of the FA surface, respectively. {For both FA deployments, $\theta_{\ell}$ and $\phi_{\ell}$ denote the polar and azimuth angles, respectively, according to the coordinate convention defined above. When $\ell=0$, these parameters correspond to the \ac{LoS} component.} Likewise, a relevant special case is the rich-scattering regime, corresponding to $K=0$ and $N_p \rightarrow \infty$ in \eqref{eq:Hkm}. Under this assumption, the channel vector associated with the $m$-th \ac{BS} antenna and user $u$ reduces to a correlated circularly symmetric complex Gaussian random vector, i.e.,
\begin{equation}
   [\mathbf{H}_u]_m \sim \mathcal{N}_\mathbb{C}(\mathbf{0},\mathbf{\Sigma}_u),
\end{equation}
where the dependence on the port locations is fully captured by the spatial correlation matrix $\mathbf{\Sigma}_u \in \mathbb{C}^{N \times N}$. In general, the structure of $\mathbf{\Sigma}_u$ is determined by the FA topology, namely, by the relative positions of the ports within the overall aperture\footnote{Specifically, for a 1D FA deployment, Jakes' correlation model is commonly adopted, whereas for a 2D FA deployment, a planar extension based on Clarke's model is considered.}. Assuming equally spaced ports, the entries of $\mathbf{\Sigma}_u$ can be compactly written as
\begin{equation}
    [\mathbf{\Sigma}_u]_{t,v}
    =
    \sigma^2
    J_0\!\left(
    2\pi d_{t,v}
    \right),
\end{equation}
where $J_0(\cdot)$ denotes the zeroth-order Bessel function of the first kind, and $d_{t,v}$ is the normalized separation between the $t$-th and $v$-th ports. For a 1D FA case, the normalized separation is given by
\begin{equation}
    d_{t,v}
    =
    \left|
    \frac{(t-v)W}{N-1}
    \right|.
\end{equation}
For a 2D FA deployment arranged on an $N_1 \times N_2$ grid, let the $t$-th and $v$-th ports be indexed by $(t_1,t_2)$ and $(v_1,v_2)$, respectively. Then, the normalized separation can be expressed as
\begin{equation}
    d_{t,v}
    =
    \sqrt{
    \left(\frac{(t_1-v_1)W_1}{N_1-1}\right)^2
    +
    \left(\frac{(t_2-v_2)W_2}{N_2-1}\right)^2
    }.
\end{equation}

Multiport FA receivers simultaneously select multiple ports, some of which are likely to be closely spaced (i.e., below $\lambda/2$). {The same consideration applies to CSI acquisition. Although perfect CSI is assumed herein, the channel over the $N$ candidate ports can be reconstructed from a reduced set of appropriately spaced samples \cite{New2025ChannelEstimation}, so the training overhead does not scale with $N$. In addition,} the effect of mutual coupling needs to be considered \cite{Nossek10}, so that the equivalent channel matrix becomes
\begin{equation}
\label{eq:coupling}
   \mathbf{H}^{\rm eq}_u = \mathbf{\Gamma}^{\rm rx}_u\mathbf{H}_u \mathbf{\Gamma}^{\rm tx},
\end{equation}
with $\mathbf{\Gamma}^{\rm tx}\in\mathbb{C}^{M\times M}$ and  $\mathbf{\Gamma}^{\rm rx}_u\in\mathbb{C}^{N\times N}$ representing the coupling matrices at the transmitter and receiver ends, respectively. Since the BS antennas have $\lambda/2$ spacing, we have $\mathbf{\Gamma}^{\rm tx}=\mathbf{I}_{M}$. Now, the structure of $\mathbf{\Gamma}^{\rm rx}_u$ heavily depends on the specific implementation\footnote{In the literature, theoretical embodiments of FAs abstract away the physical implementation behind the FA concept. In practice, spatial flexibility can be achieved mechanically \cite{Zhu2024, Shen2024}, or electronically through pixel-based \cite{Zhang2024,Zhang2025} or metasurface-based \cite{Ramirez2025,Liu2025} embodiments.} for the FA. Hence, the coupling matrix is often extracted from full-wave simulators \cite{Ramirez2025}, experimentally, or approximated by modeling each port as a dipole of length $\lambda/2$ and minimal width, with a proper termination impedance (classically $50\Omega$) \cite{Dinis26}.

Under the assumption that the coupling matrix $\mathbf{\Gamma}^{\rm rx}_u$ does not affect $\mathbf{n}_u$ in \eqref{eq1} \cite{CUMA}, all signal processing procedures in FAHM (and the competing benchmark schemes) are applicable only replacing $\mathbf{H}_u$ by $\mathbf{H}^{\rm eq}_u$ from \eqref{eq:coupling}. However, in some circumstances the mutual coupling may induce noise correlation at the receiver side \cite{Nossek10}. This has an impact on the eigenvector structure discussed later Section \ref{subs:FAHM}, and needs a careful examination which is beyond the scope of this paper.

\vspace{0mm}
\section{Receiver Design}
\label{RD}
In this section, we detail the receiver-side design of the proposed FAHM architecture. We first characterize the hybrid multiport receiver for a fixed selected-port set and show how the analog and digital stages can be represented through an effective combiner. Next, we present the GEPort-based selection procedure together with its reduced-complexity {implementation} and the effective-port criterion used to determine the dimension {of the set of selected ports}. Finally, we show that the {reference} CUMA receiver \cite{CUMA} can be recovered as a particular structured instance of the proposed FAHM formulation.

\subsection{Hybrid Multiport Receiver }
\label{subs:FAHM}
 On the receiver side, each FA-equipped user determines which $P$ ports to activate, as well as the analog network weights and the digital combining vector, according to its channel conditions and the level of interference. The first stage is represented by the port-selection matrix $\mathbf{S}_u \in \mathcal{B}$, where {the admissible set is defined as}
\begin{equation}
{
\mathcal{B}
:=
\left\{
\mathbf{Z}\in\{0,1\}^{N\times P}\; \middle|\;
\begin{aligned}
&\sum_{i=1}^{N} Z_{i,p}=1,
&& \forall p=1,\ldots,P,\\
&\sum_{p=1}^{P} Z_{i,p}\leq 1,
&& \forall i=1,\ldots,N
\end{aligned}
\right\}.
}
\end{equation}
{The first constraint ensures that each column of $\mathbf{Z}$ selects exactly one FA port, whereas the second constraint prevents the same port from being selected more than once.} Furthermore, the signals collected from the selected $P$ ports can be coherently combined through the analog combiner $\mathbf{F}_u \in \mathcal{F}$, with $\mathcal{F}:=\left\{ \mathbf{Z}\in\mathbb{C}^{P\times L} : |\mathbf{Z}_{i,j}| = 1,\; \forall i,j \right\}$, together with the digital combining vector $\mathbf{w}_u \in \mathbb{C}^{L}$, which satisfies $\|\mathbf{w}_u\|_2 = 1$ for all $u$. Under this receiver processing, the detected symbol at user $u$ can be expressed as
\begin{equation}
    \hat{z}_u = \mathbf{w}_u^H \mathbf{F}_u^H \mathbf{S}_u^T \mathbf{x}_u,
\end{equation}
where $\mathbf{x}_u$ is given in \eqref{eq1}. Accordingly, the \ac{SINR} of user $u$ reads as
\begin{equation}
    \text{SINR}_u = \frac{|\mathbf{w}_u^H \mathbf{F}_u^H\mathbf{S}_u^T\mathbf{H}_u\mathbf{p}_u|^2}{\sum_{j \neq u} |\mathbf{w}_u^H \mathbf{F}_u^H\mathbf{S}_u^T\mathbf{H}_u\mathbf{p}_j|^2 + \|\mathbf{F}_u \mathbf{w}_u\|_2^2\frac{1}{\mathrm{SNR}}},
    \label{eq:SINR}
\end{equation}
with $\mathrm{SNR}=\sigma_S^2/\sigma^2$ being the transmit SNR. Notice that the joint design of $\mathbf{S}_u$, $\mathbf{F}_u$, and $\mathbf{w}_u$ that optimizes the SINR is challenging. Therefore, we first define a hybrid analog--digital combiner through the complex vector $\mathbf{t}_u\in \mathbb{C}^{P}$, 
\begin{equation}
\label{eqtu}
    \mathbf{t}_u = \mathbf{F}_u \mathbf{w}_u .
\end{equation}
Once $\mathbf{t}_u$ is determined, it {is subsequently} decomposed into $\mathbf{F}_u$ and $\mathbf{w}_u$, while preserving the receiver structure. This reduces the problem at hand to the design of the selection matrix $\mathbf{S}_u$ and an equivalent combining vector $\mathbf{t}_u$. Still, the optimization problem is difficult due to the presence of the variables in the desired signal and the interference terms. {One alternative is to perform a suboptimal select-then-combine strategy, where the number of selected ports is first made according to some criterion (e.g., those maximizing the per-port SINR, as in the DC strategy), and then the combining vector is designed accordingly. A joint select-and-combine strategy was shown to provide nearly-optimal performance under a fully-digital combining architecture \cite{GoLo26}. In this case, adapting the problem formulation to the FAHM case, for} a given selection matrix $\mathbf{S}_u$ the equivalent combiner $\mathbf{t}_u$ is obtained as the dominant generalized eigenvector (GEV) of the matrix pair $(\tilde{\mathbf{A}}_u,\tilde{\mathbf{B}}_u)$ \cite{GoLo26}, where
\begin{equation}
\tilde{\mathbf{A}}_u=\mathbf{S}_u^T\mathbf{A}_u\mathbf{S}_u,
\qquad
\tilde{\mathbf{B}}_u=\mathbf{S}_u^T\mathbf{B}_u\mathbf{S}_u,
\label{eq:AtildeBtilde_u}
\end{equation}
with
\begin{equation}
\mathbf{A}_u=\mathbf{H}_u\mathbf{p}_u\mathbf{p}_u^H\mathbf{H}_u^H,
\qquad
\mathbf{B}_u=\sum_{j\neq u}\mathbf{H}_u\mathbf{p}_j\mathbf{p}_j^H\mathbf{H}_u^H+\frac{\mathbf{I}_N}{\mathrm{SNR}}.
\label{eq:AuBu_u}
\end{equation}
Hence, the receiver design over the selected-port subspace can be carried out directly in terms of $\mathbf{t}_u$. Once this dominant GEV is computed, a feasible hybrid realization is obtained by decomposing $\mathbf{t}_u$ into the analog combiner $\mathbf{F}_u$ and the digital vector $\mathbf{w}_u$. This reformulation is very relevant in the considered slow-FAMA setting, where each user detects a single data stream. In such a scenario, the RF/baseband codesign result in \cite{ZhMoKu05} implies that the same optimal linear combining gain as in the corresponding digital implementation can be attained with \textit{only two RF chains} over a fixed selected-port subspace. Therefore, in the proposed FAHM receiver, it is sufficient to set $L=2$ in $\mathbf{F}_u$ in order to realize the effective optimal combiner $\mathbf{t}_u$ over the selected ports, while preserving the hybrid receiver structure. Since FAMA operates in an open-loop fashion, each user can independently design its corresponding $\mathbf{S}_u$ and $\mathbf{t}_u$. 

{The known two-RF-chain decomposition is therefore an enabling implementation mechanism, more than an analog-beamforming novelty. Unlike conventional hybrid MIMO with a fixed array, FAHM coordinates the channel-adapted selection stage $\mathbf{S}_u$ with $\mathbf{F}_u$ and $\mathbf{w}_u$ under $N\geq P\geq L$; the effective-port criterion, port-subset design, and reduced-complexity implementation define the FAHM contribution.}

\subsection{Complexity Reduction and Effective-Port Criterion}
The design of the port-selection matrix $\mathbf{S}_u$ in the proposed hybrid receiver is based on the GEPort method introduced in \cite{GoLo26}. This strategy relies on backward elimination, sequentially discarding the least relevant ports according to the induced performance degradation. Therefore, in the FAHM receiver, GEPort is used to determine the selected-port subset that will subsequently feed the hybrid analog--digital combiner. Specifically, following Algorithm~1 in \cite{GoLo26}, the selection procedure iteratively removes the least relevant port from the current candidate set until the desired selected-port dimension is attained. Although this approach provides strong performance, its original realization requires repeated generalized eigendecompositions on progressively smaller problems, leading to a significant computational cost. This burden is considerably alleviated in the proposed FAHM implementation, while preserving the same GEPort decision rule. {Specifically, in the proposed FAHM architecture, we retain the GEPort ranking mechanism but introduce two complementary improvements: a reduced-complexity implementation and a channel-dependent criterion for determining the number of ports that should be preserved before hybrid combining. We first address the former by developing a
reduced-complexity realization of GEPort, as stated in the following
proposition:} 
\begin{proposition}
The original GEPort implementation in \cite{GoLo26} requires repeatedly solving a generalized eigenvalue problem, which leads to a  {worst-case} per-iteration complexity of $\mathcal{O}(N^3)$. In the proposed FAHM realization, this cost can be reduced to $\mathcal{O}(N^2)$ per iteration by exploiting the rank-1 structure of the useful-signal matrix and recursively updating the inverse of the interference-plus-noise matrix. The GEPort decision rule itself remains unchanged; only its implementation is accelerated. 
\end{proposition}
\begin{proof}
See Appendix \ref{apA}.
\end{proof}
For the original multiport receiver in \cite{GoLo26}, {which provides a fully-digital solution, }the number of selected ports is fixed {and equal to the number of RF chains. For practical reasons, this number should be kept as low as possible to reduce complexity. Conversely, the} proposed FAHM architecture {allows to effectively decouple the dimension of the selected port set and the number of RF chains. Hence, a key design question needs to be solved for a FAHM receiver:} \textit{how many ports should be activated to perform hybrid combining?} {Notably, this allows to overcome one of the key limitations of} GEPort: owing to the aggressive sequential dimensionality reduction, even quasi-optimal early-stage port removal decisions {cause a} substantial SINR degradation in later iterations, since the spatial information, primarily exploited during the initial iterations, is progressively discarded as the matrix dimension is reduced.

{To provide a quantitative mechanism that captures the SINR degradation vs. the number of discarded ports}, we define {the \textit{effective number of selected ports}}, denoted by $P_{\mathrm{eff}}$ as the minimum number of active ports required to preserve a satisfactory SINR level and, therefore, to retain the dominant contribution of the channel in the combining process. {In particular, $P_{\mathrm{eff}}$ measures how broadly the dominant combining direction is distributed across the FA aperture and provides a channel-dependent criterion for determining the dimension of the port subspace that should be preserved before hybrid combining. In this way, the proposed FAHM architecture avoids forcing the selected-port dimension to coincide with the RF-chain budget, while still preserving the GEPort relevance ordering for identifying the final subset of ports. We next formalize this second improvement through the effective-port
criterion stated in the following proposition.}
\begin{proposition}
\label{prop2}
Let $\mathbf{v}=[v_1,\ldots,v_N]^T \in \mathbb{C}^{N}$ denote
the dominant GEV associated with the full $N$-port candidate set,
normalized such that $\|\mathbf{v}\|_2=1$. Then, the effective
number of selected ports is given by
\begin{equation}
    P_{\mathrm{eff}}=\frac{1}{\sum_{i=1}^{N}|v_i|^4}.
\end{equation}
\end{proposition}

\begin{proof}
See Appendix~\ref{apB}.
\end{proof}
{By virtue of Proposition \ref{prop2}}, in the proposed FAHM-GEPort receiver the selected-port dimension is chosen as
\begin{equation}
\label{Peff}
    P^\star = \left\lceil P_{\mathrm{eff}} \right\rceil,
\end{equation}
where $\lceil \cdot \rceil$ denotes the ceiling operator,
and GEPort is applied until only $P^\star$ ports remain in the selected set. {For clarity, the complete FAHM port-selection workflow is illustrated in Fig.~\ref{figPeff}. Starting from the full $N$-port candidate set, the normalized dominant GEV is first used to evaluate $P_{\mathrm{eff}}$, after which the target selected-port dimension is set as $P^\star=\lceil P_{\mathrm{eff}}\rceil$. The GEPort backward-elimination procedure is then applied until exactly $P^\star$ ports remain, and the resulting subset is finally used to construct $\mathbf{S}_u$ and design the hybrid analog--digital combiner.}
\begin{figure}[!t]
\includegraphics[width=0.95\linewidth]{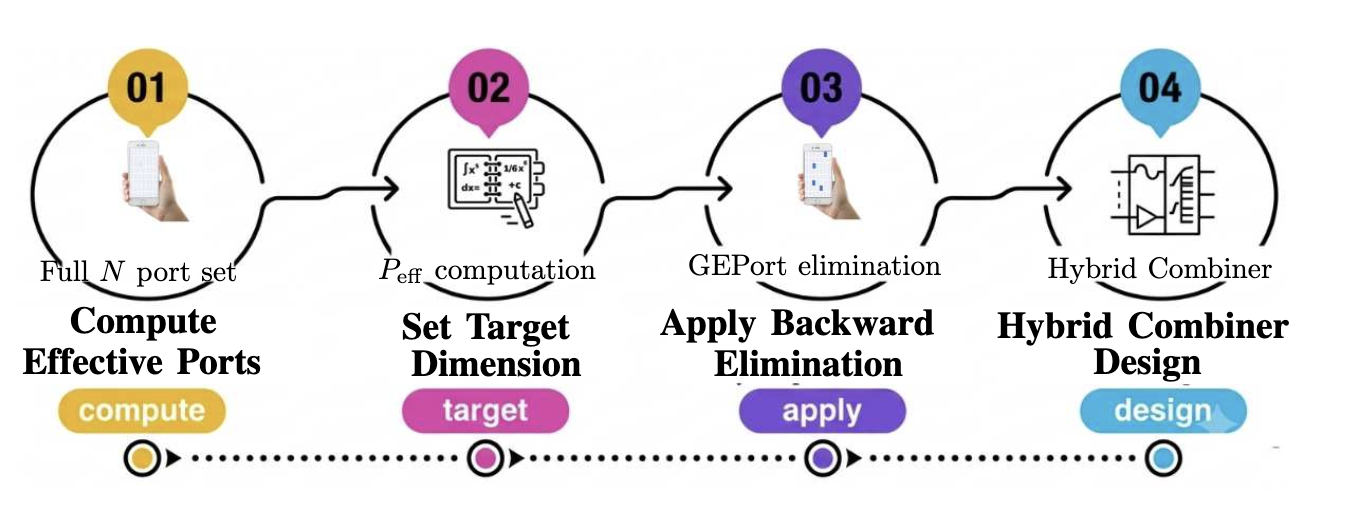}
\vspace{0mm}\caption{{Illustration of the proposed FAHM port-selection and hybrid-combining workflow.}}
\label{figPeff}
\vspace{0mm}\end{figure}
{Hence, the proposed effective-port criterion determines \emph{how many} ports should be preserved, whereas GEPort determines \emph{which} $P^\star$ ports are retained. This separation is central in the proposed FAHM design, since it decouples the selected-port dimension from the RF-chain budget and allows the hybrid combiner to operate on a channel-adapted retained-port subspace.}
\begin{remark}
    The quantity $P_{\mathrm{eff}}$ acts as an inverse concentration measure of the dominant GEV. Small values indicate that the dominant combining contribution is concentrated on only a few FA ports, whereas larger values indicate that it is distributed over a broader portion of the aperture. Therefore, $P_{\mathrm{eff}}$ provides a channel-dependent estimate of the port-subspace dimension that should be retained before hybrid combining. Once $P^\star=\lceil P_{\mathrm{eff}}\rceil$ is determined, GEPort is used to identify the corresponding $P^\star$ most relevant ports.
\end{remark}

\subsection{CUMA as a Special Case of FAHM}

The generality of the proposed FAHM structure also allows existing multiport receivers to be interpreted as particular cases. Specifically, the CUMA architecture can be recovered by imposing a structured analog network and a suitable partition of the selected ports into in-phase and quadrature subsets, as stated in the following proposition.

\begin{proposition}
Let $\mathcal{P}_I$ and $\mathcal{P}_Q$ denote the subsets associated with the in-phase and quadrature components of the channel, respectively. Then, by setting $L=4$ and properly choosing the matrices $\mathbf{S}_u$, $\mathbf{F}_u$, and the vector $\mathbf{w}_u$, the proposed FAHM receiver recovers CUMA \cite{CUMA} as a particular structured instance. In particular, the hybrid output can be expressed as
\begin{equation}
    \mathbf{r}_u=\mathbf{F}_u^H\mathbf{S}_u^T\mathbf{x}_u,
\end{equation}
and the detected symbol is obtained as
\begin{equation}
\hat{z}_u=\mathbf{w}_u^H\mathbf{r}_u.
\end{equation}
with $\mathbf{w}_u^H$ given by \eqref{eqWcuma} in Appendix~\ref{apC}.
\end{proposition}

\begin{proof}
See Appendix~\ref{apC}.
\end{proof}
\begin{remark}
The above result highlights that the proposed FAHM formulation is not limited to a single combining strategy, but rather provides a unifying framework that includes conventional multiport architectures as special instances, with CUMA arising as a {particular structured case with $L=4$}. {We note that the digital combining vector $\mathbf{w}_u\in\mathbb{C}^{4\times 1}$ for notational consistence, due to the inability to express the $\Re/\Im$ operations as complex linear operators. In practice, CUMA can be implemented with only 2 RF chains, admitting an extension to 4 RF chains \cite{CUMA2}.} 
\end{remark}

\vspace{0mm}
\section{Performance Metrics and Numerical Results}
\label{SecIV}
%
\subsection{Performance Metrics}
\label{pm}
To evaluate the proposed FAHM receiver, we consider two complementary user-wise performance metrics, namely, the \ac{SE} and the \ac{OP}. Both metrics are directly determined by the post-combining \ac{SINR} achieved by the hybrid receiver. For the adopted FAHM structure, the SINR of user $u$ is given by \eqref{eq:SINR}. The SE achieved\footnote{{In practice, the training overhead required for channel estimation affects the SE in FAMA as a pre-log factor $\eta=1-\frac{N_S U}{L T_c}$, where $T_c$ is the channel coherence time (in symbol periods) and $N_S\ll N$ is the number of sampled/estimated port locations. Since $\eta$ is generally common to all considered FAMA schemes, it does not alter their relative performance ranking. Hence, $\eta=1$ is assumed in the subsequent evaluations of $R_u$.}} by user $u$ is therefore given by
\begin{equation}
R_u={\eta}\log_2\!\left(1+\mathrm{SINR}_u\right).
\label{eq:SE_metric}
\end{equation}
Accordingly, an SE-oriented formulation of the receiver design can be written as
\begin{equation}
\max_{\{\mathbf{S}_u\in\mathcal{B},\,\mathbf{t}_u\}_{u=1}^{U}}
\sum_{u=1}^{U}\log_2\!\left(1+\mathrm{SINR}_u\right).
\label{eq:SE_optimization}
\end{equation}
In addition, let $P_{\text{out},u}(\gamma)$ denote the \ac{OP} of user $u$, where $\gamma$ is the target SE threshold. Then, the corresponding reliability-oriented formulation can be stated as
\begin{align}
\underset{\mathbf S_u\in\mathcal{B},\mathbf t_u}{\text{min}}
P_{\text{out},u}(\gamma), \;\;\text{s.t.} \;\;P_{\text{out},u}(\gamma)=P(\log_2(1+\text{SINR}_u)<\gamma). 
\label{eq:probForm}
\end{align}

\subsection{Numerical Results}
\label{nr}
We now evaluate the performance of the proposed FAHM receiver architecture for multiport FAMA operation. {As previously discussed, in} the hybrid multiport receiver the number of selected active ports is not necessarily equal to the number of RF chains. Specifically, $P$ ports are activated at the FA, whereas only $L$ RF chains are employed after analog combining, typically with $P\gg L$. Unless otherwise stated, and {due to the optimality of using only two RF chains over a fixed selected-port subspace \cite{ZhMoKu05},} all figures for the proposed FAHM scheme are obtained with $L=2$ RF chains. {Without loss of generality}, all plots are generated by assuming unit average transmit-symbol power, i.e., $\sigma_S^2 = 1$. {In the channel generation, the \ac{LoS} components observed at the FA receiver are modeled with a common angle of arrival for all antenna-user pairs, namely, $\theta_{0}^{(u,m)}=\pi/2$ and $\phi_{0}^{(u,m)}=0$. By contrast, the angles of the scattered paths, $\{(\theta_{\ell}^{(u,m)},\phi_{\ell}^{(u,m)})\}_{\ell=1}^{N_p}$, are independently generated according to the adopted isotropic angular distribution. Likewise, we generate $\delta_{u,m}\sim\mathcal{U}[0,2\pi)$ and $\kappa_{\ell}^{(u,m)}\sim\mathcal{N}_{\mathbb{C}}(0,1)$, independently across paths and antenna-user pairs.} For comparison purposes, we benchmark the proposed FAHM receiver against slow-FAMA (single RF chain), CUMA-based reception{\footnote{{CUMA requires a heuristic tuning of its operational parameters, namely the maximum number of selected ports $N_{\max}$ and the threshold parameter $\rho\in[0,1]$. In all subsequent comparisons, we use the values $N_{\max}=N$ and $\rho=0.4$ since they provided the best performance after an exhaustive search over the candidate configurations considered
in our simulations. Hence, the selected number of ports in CUMA largely exceeds that of $P$ in FAHM.}}} ($L=2$ RF chains), and FAHM-DC architecture {with equal $P$ and $L=2$.} {Unless otherwise stated, $\mathbf{\Gamma}^{\rm rx}_u=\mathbf{I}_N$ is assumed for all compared schemes.} {Across this section, we include different representative configurations for FA geometries (elongated, rectangular and square). These configurations are used to evaluate the receiver schemes under different representative FA geometries. They are not intended to isolate the effect of aperture size, since other simulation parameters also vary across the figures. Within each figure, all compared schemes use the same aperture and system parameters.}

\begin{figure}[t]
    \centering
\includegraphics[width=.98\columnwidth]{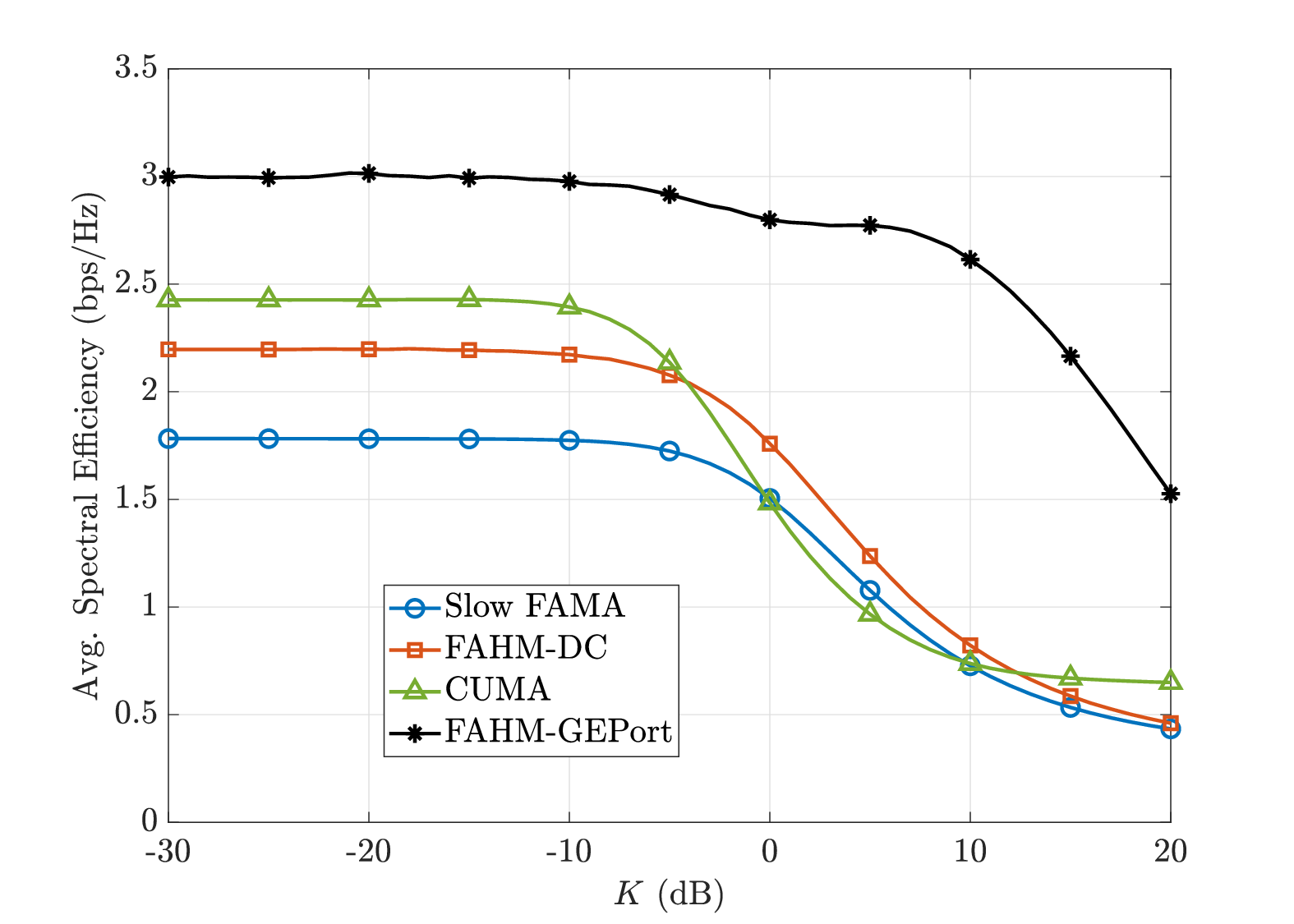}
\vspace{0mm}\caption{Average SE for user {$u$} versus the Rice factor $K$ for the different schemes. Results correspond to the 2D FA case with $N_1 \times N_2 = 15 \times 15$ and $W_1 \times W_2 = 4 \times 1$. Also, {$M=U=5$}, $P=2$, $N_p=30$, and $\mathrm{SNR}=15$ dB.}
    \label{fig2}
\vspace{0mm}\end{figure}

We first evaluate the performance of FAHM for some specific values of $P$, assuming different propagation conditions. In Fig.~\ref{fig2}, we explore the impact of the Rice factor $K$ on the achievable average spectral efficiency (SE) of the considered schemes. The results correspond to the 2D FA case and are obtained for a multiuser setup with $M=U=5$, $P=2$, a $N=15\times 15$ FA grid, and a normalized aperture $W= 4 \times 1$, with $N_p=30$ scattered paths, and $\mathrm{SNR}=15$ dB. {Note that since $P<U$ under this configuration, the system is operating in a regime where interfering signals exceed the available degrees of freedom, and hence interference rejection capabilities are severely limited. }As observed in the all curves, GEPort consistently achieves the highest average SE over the entire range of Rice factors under consideration, which highlights the benefit of jointly exploiting port selection and combining. FAHM-DC and CUMA provide intermediate performance, whereas conventional slow FAMA remains the weakest benchmark. Interestingly, the performance gap between FAHM-GEPort and the reference schemes is especially clear in the low- and moderate-$K$ regimes, where the channel still preserves a significant scattered component. By contrast, for very large $K$, all schemes experience a marked performance degradation. In the considered setup, this behavior can be attributed to the increasing dominance of the common LoS component, which reduces angular diversity and makes user separation more challenging under the open-loop slow-FAMA framework. Nevertheless, FAHM-GEPort remains the most robust strategy across all values of $K$.

\begin{figure}[t]
    \centering
\includegraphics[width=.98\columnwidth]{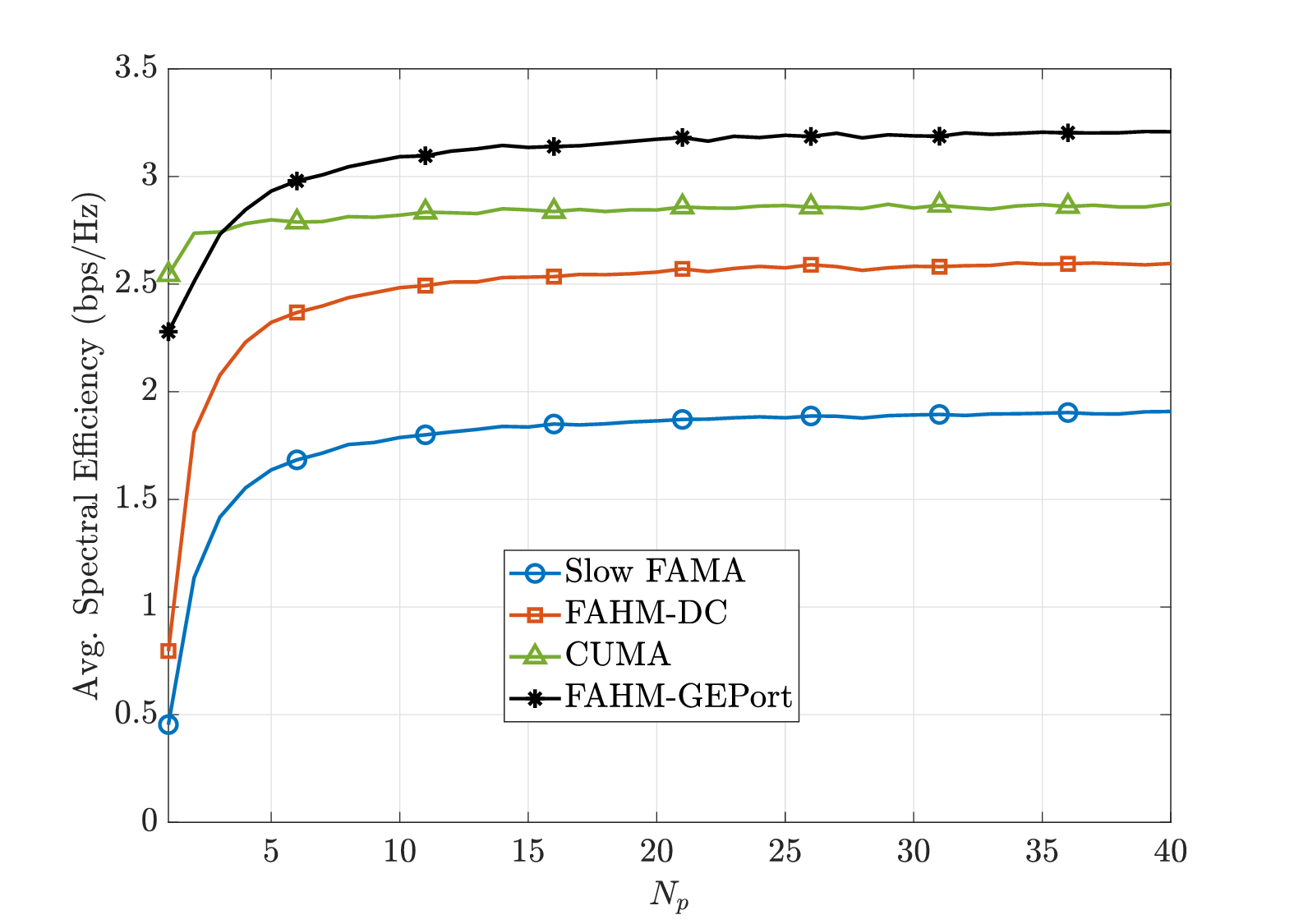}
\vspace{0mm}\caption{Average SE for user $u$ versus the number of scattered paths $N_p$ for the different schemes. Results correspond to the 2D FA case with $N_1 \times N_2 = 15 \times 15$ and {$W_1 \times W_2 = 5 \times 4$}. Also, $M=U=5$, $P=2$, $K=-20$ dB, and $\mathrm{SNR}=10$ dB.}
    \label{fig3}
\vspace{0mm}\end{figure}

\begin{figure}[t]
    \centering
\includegraphics[width=.98\columnwidth]{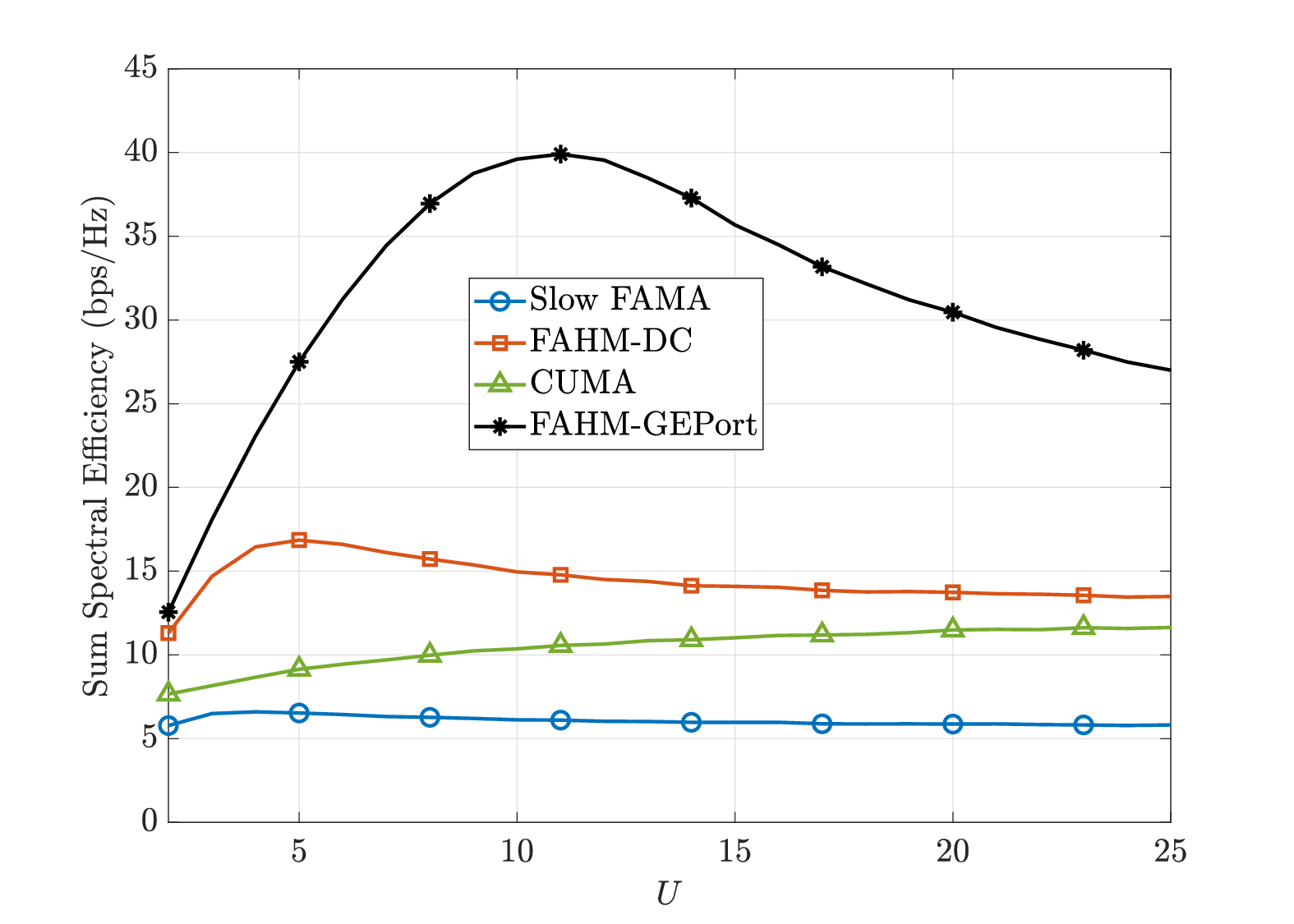}
\vspace{0mm}\caption{{Sum} SE versus the number of BS antennas and users, with $M=U$, for the different schemes. Here, we assume 1D FA case with $N=100$ ports and normalized aperture $W=6$. Also, $P=10$ and $\mathrm{SNR}=5$ dB.}
    \label{fig4}
\vspace{0mm}\end{figure}

In Fig.~\ref{fig3}, we investigate 
the impact of the number of scattered paths $N_p$ on the average SE of the considered schemes. In the multiuser setup, we set: $M=U=5$, $P=2$, a $N= 15\times 15$ FA grid, and a normalized aperture {$W=5\times 4$}, while the Rice factor is fixed to $K=-20$ dB and the transmit SNR is set to $10$ dB. From all instances, it can be observed that the average SE improves as $N_p$ increases, especially when moving from the extremely sparse-scattering regime to moderate values of $N_p$. This behavior indicates that a richer multipath environment provides additional spatial diversity across the FA ports, which can be more effectively exploited by port-selection and combining strategies. Among all schemes, FAHM-GEPort achieves the highest average SE over the entire range of $N_p$, confirming the benefit of jointly exploiting multiport selection and hybrid combining. CUMA and FAHM-DC provide intermediate performance, whereas conventional slow-FAMA remains the weakest benchmark. In particular, the advantage of FAHM-GEPort is especially noticeable in the low-$N_p$ regime, where the channel offers limited angular richness and an efficient exploitation of the most informative ports becomes more critical. Moreover, although all schemes benefit from increasing $N_p$, the gains gradually saturate for large values of $N_p$. 

In Fig.~\ref{fig4}, we study the impact of the number of users $U$ on the {sum} spectral SE achieved by the considered schemes. In this setup, we assume a 1D-FA case with $N=100$ ports, $P=10$ selected ports, a normalized aperture $W=6$, and Rayleigh channel fading, while the transmit SNR is fixed to $5$ dB. As observed from all curves, {the sum SE does not decrease monotonically with \(U\); instead, the different schemes exhibit distinct behaviors as the number of simultaneously served users increases.} For instance, in the considered slow-FAMA setting, increasing \(U\) also increases the number of simultaneously served users, which strengthens the multiuser interference term in the SINR expression. In particular, although a larger number of BS antennas may in principle provide additional spatial degrees of freedom, the open-loop transmission strategy considered here does not exploit transmit-side CSI, and therefore {the achievable gain is governed by the trade-off between the additional receiver-side diversity and the increased multiuser interference.} Among all schemes, FAHM-GEPort achieves the highest SE over the entire range of \(M\), showing a clear advantage in the low- and moderate-\(U\) regimes. {For FAHM-GEPort, the sum SE initially increases with \(U\), reaching its maximum at a moderate number of users (comparable to $P$), and then progressively decreases as the interference load becomes dominant. A similar but less pronounced behavior is observed for FAHM-DC, whereas CUMA and slow FAMA saturate much earlier, and then become almost insensitive to \(U\).} {In the high-\(U\) regime, the performance gain of FAHM-GEPort is gradually reduced,} which indicates that the benefit of port selection and combining becomes increasingly limited when the interference load scales together with the number of users.


In Fig.~\ref{fig5}, we examine the impact of the number of selected ports, denoted by $P$, on the average SE achieved by the considered schemes. In this setup, we consider a multiuser 2D FA scenario with $M=U=6$, $N = 15 \times 15$, and normalized aperture {$W = 5 \times 5$}, while the Rice factor is fixed to $K=10$ dB, the number of scattered paths is set to $N_p=80$, and the transmit $\mathrm{SNR}=10$ dB. As observed from the plots, the average SE increases with the number of selected ports for the FAHM-GEPort and FAHM-DC schemes, {providing remarkable improvements over the slow FAMA and CUMA benchmarks.} In particular, FAHM-GEPort achieves the highest SE over the entire range of selected ports, clearly outperforming the reference schemes. {Interestingly,} the performance improvement is not uniform over the whole range of $P$. Most of the SE gain is obtained in the low- and moderate-$P$ regime, where increasing the number of selected ports from very small values produces a sharp improvement, especially for FAHM-GEPort. By contrast, for large $P$, the curves progressively flatten, revealing a clear effect of diminishing-returns. This saturation indicates that once a sufficiently rich subset of ports has already been selected, adding more ports provides only marginal additional diversity or interference mitigation. {With this in mind, we further proceed to formally analyze the design of $P$ using the concept of effective number of ports $P_{\rm eff}$, and to showcase its usefulness for system design in FAHM}.

\begin{figure}[!t]
    \centering
\includegraphics[width=.98\columnwidth]{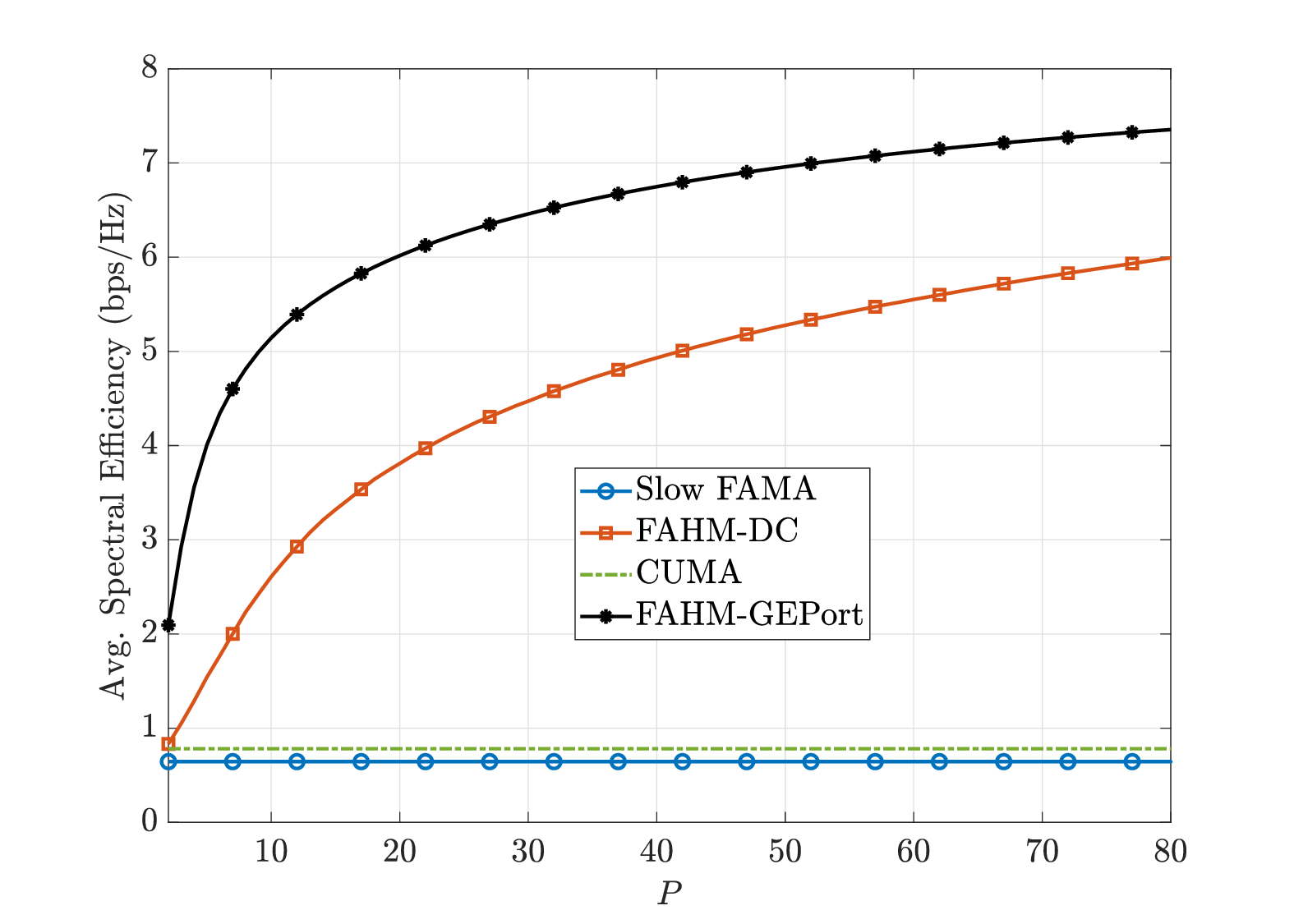}
\vspace{0mm}\caption{Average SE for user $u$ versus the number of selected ports $P$ for the different schemes. Here, we assume a 2D FA case with $N_1 \times N_2 = 15 \times 15$ and normalized aperture {$W_1 \times W_2 = 5 \times 5$}. Also, $M=U=6$, $K=10$ dB, $N_p=80$, and $\mathrm{SNR}=10$ dB.}
    \label{fig5}
\vspace{0mm}\end{figure}

\begin{figure*}[ht!]
\centering
\subfigure[Average SE for user $u$ vs. SNR for $P=30$.]{\includegraphics[width=.48\textwidth]{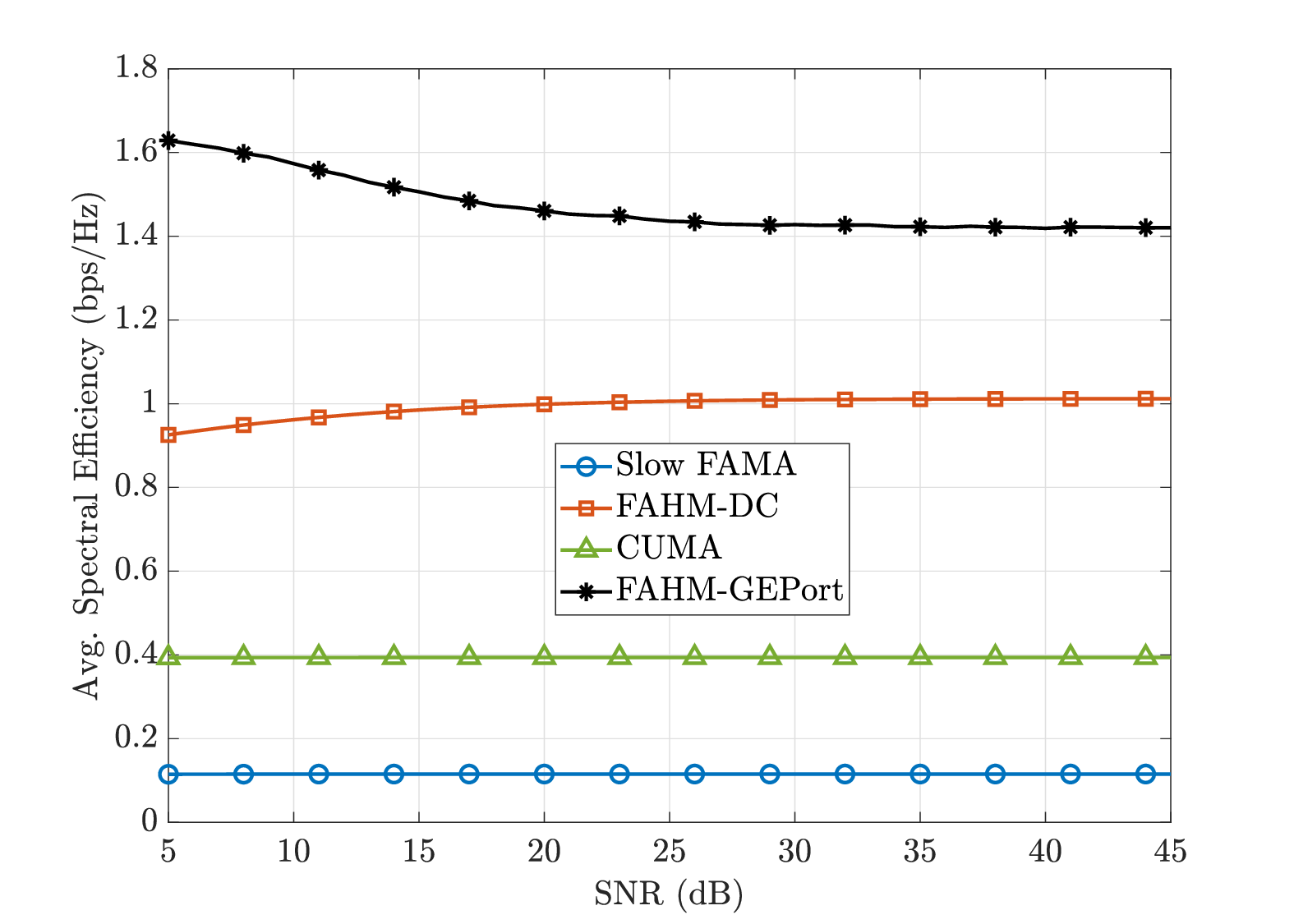}} 
\subfigure[Average SE for user $u$ vs. SNR for $P=60$. ]{\includegraphics[width=.48\textwidth]{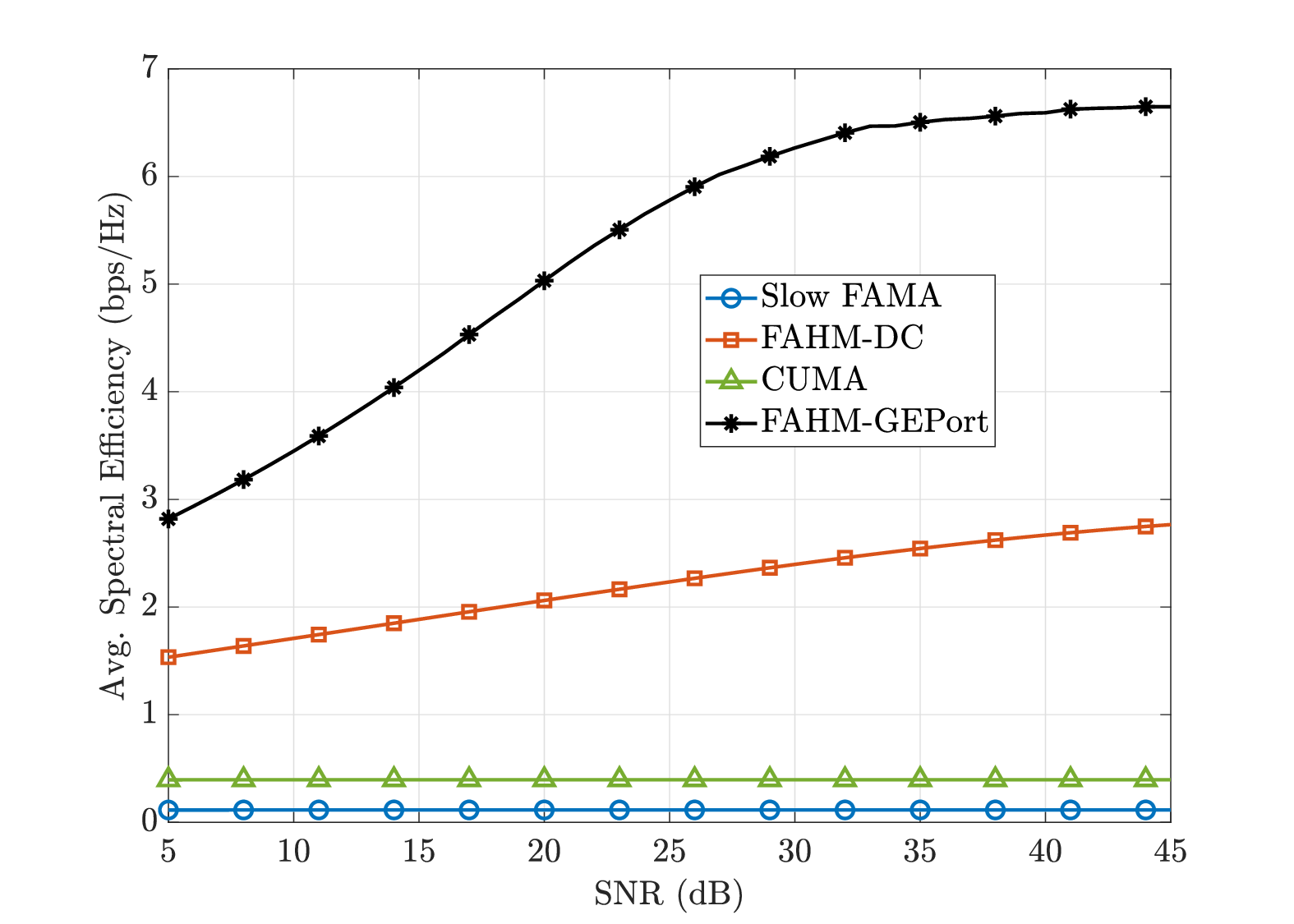}}
\subfigure[Average SE for user $u$ vs. SNR for {$P=P_{\textit{\rm eff}}=85$}.]{\includegraphics[width=.48\textwidth]{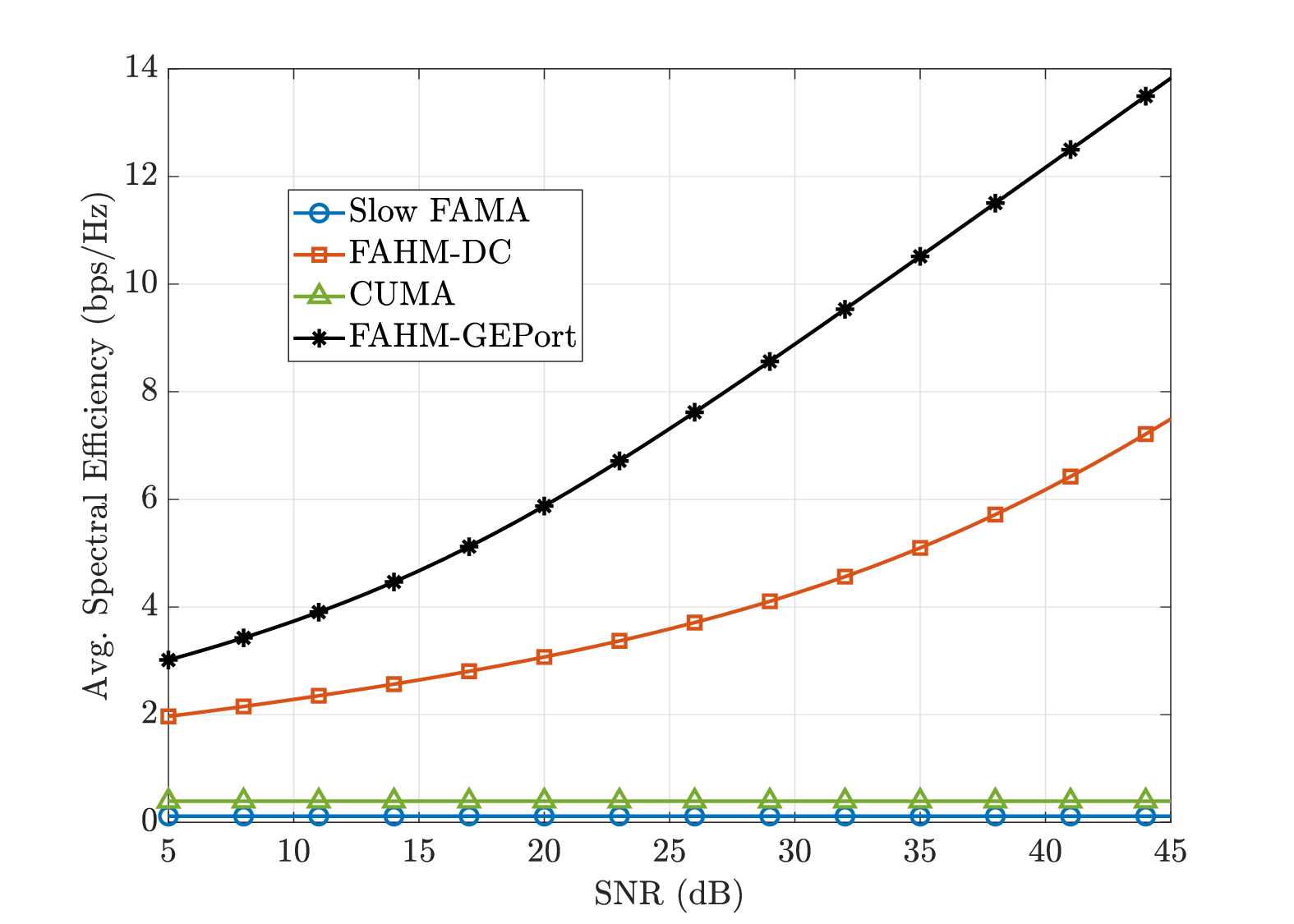}}
\subfigure[Dominant-SINR elbow analysis for $P_{\mathrm{eff}}$.]{\includegraphics[width=.48\textwidth]{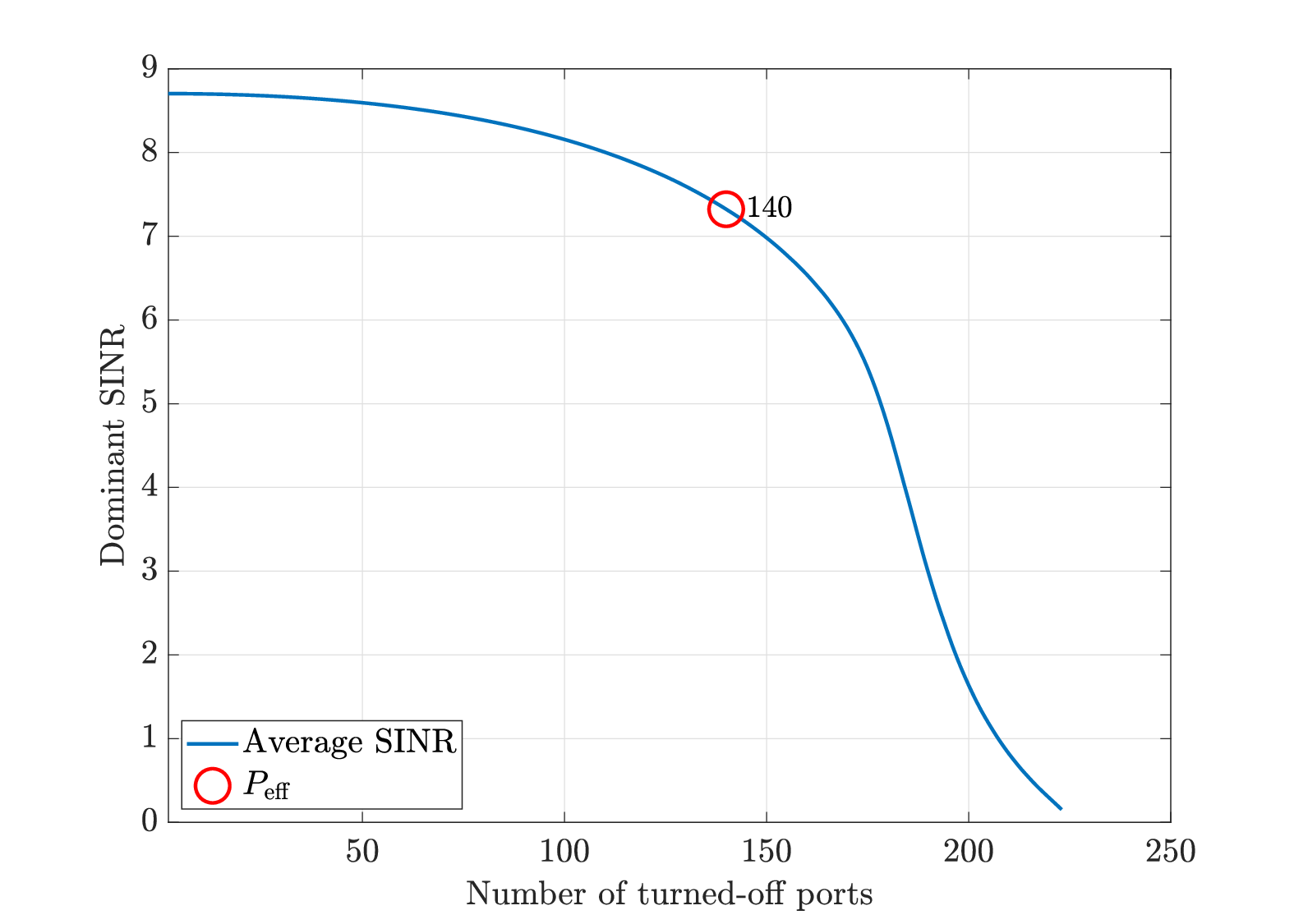}}
\vspace{0mm}\caption{Average SE for user $u$ versus the transmit SNR for the different schemes and different selected-ports, namely, Fig.~\ref{fig6}(a) for
 $P=30$, Fig.~\ref{fig6}(b) for  $P=60$, and Fig.~\ref{fig6}(c) for {$P=85$}. On the other hand, Fig.~\ref{fig6}(d) shows the elbow-based dominant-SINR analysis used to estimate the effective number of ports. Here, we assume a 2D FA case with $N_1\times N_2=15\times 15$, so that $N=225$, and $P=85$ is obtained from the effective-port indicator $P_{\mathrm{eff}}$.}
    \label{fig6}
\end{figure*}


\begin{figure}[!t]
    \centering
\includegraphics[width=.98\columnwidth]{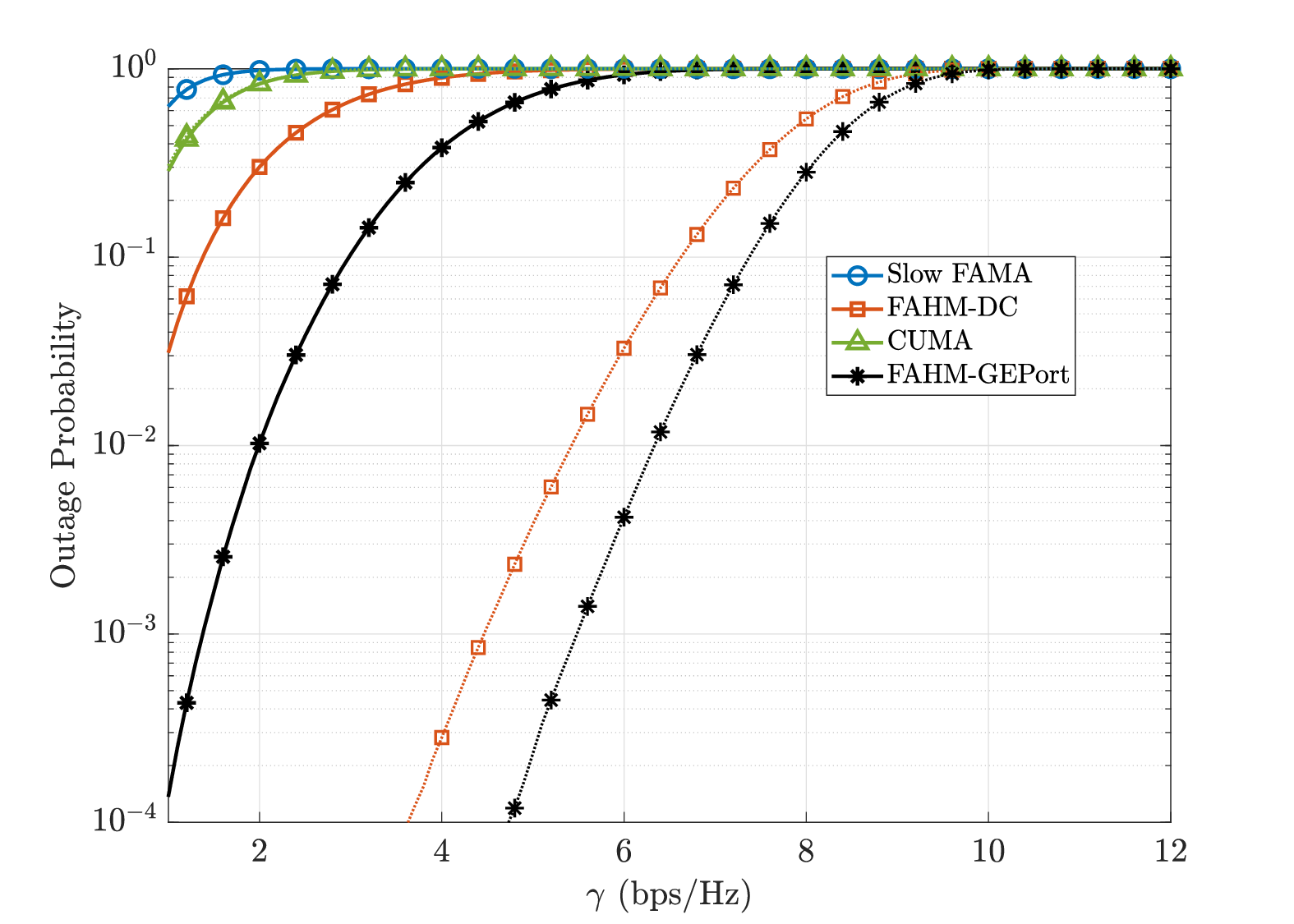}
\vspace{0mm}\caption{Outage probability for user $u$ vs. $\gamma$ for the different schemes. Here, we assume a 2D FA case with $N = 5 \times 5$ and normalized aperture {$W= 2 \times 1$.} Also, $M=U=6$, $K=-10$ dB, $N_p=5$, and $\mathrm{SNR}=15$ dB. Solid and dotted lines correspond to $P=4$ and $P=P_{\mathrm{eff}}=14$, respectively.}
    \label{fig8}
\vspace{0mm}\end{figure}

\begin{figure}[!t]
    \centering
\includegraphics[width=.98\columnwidth]{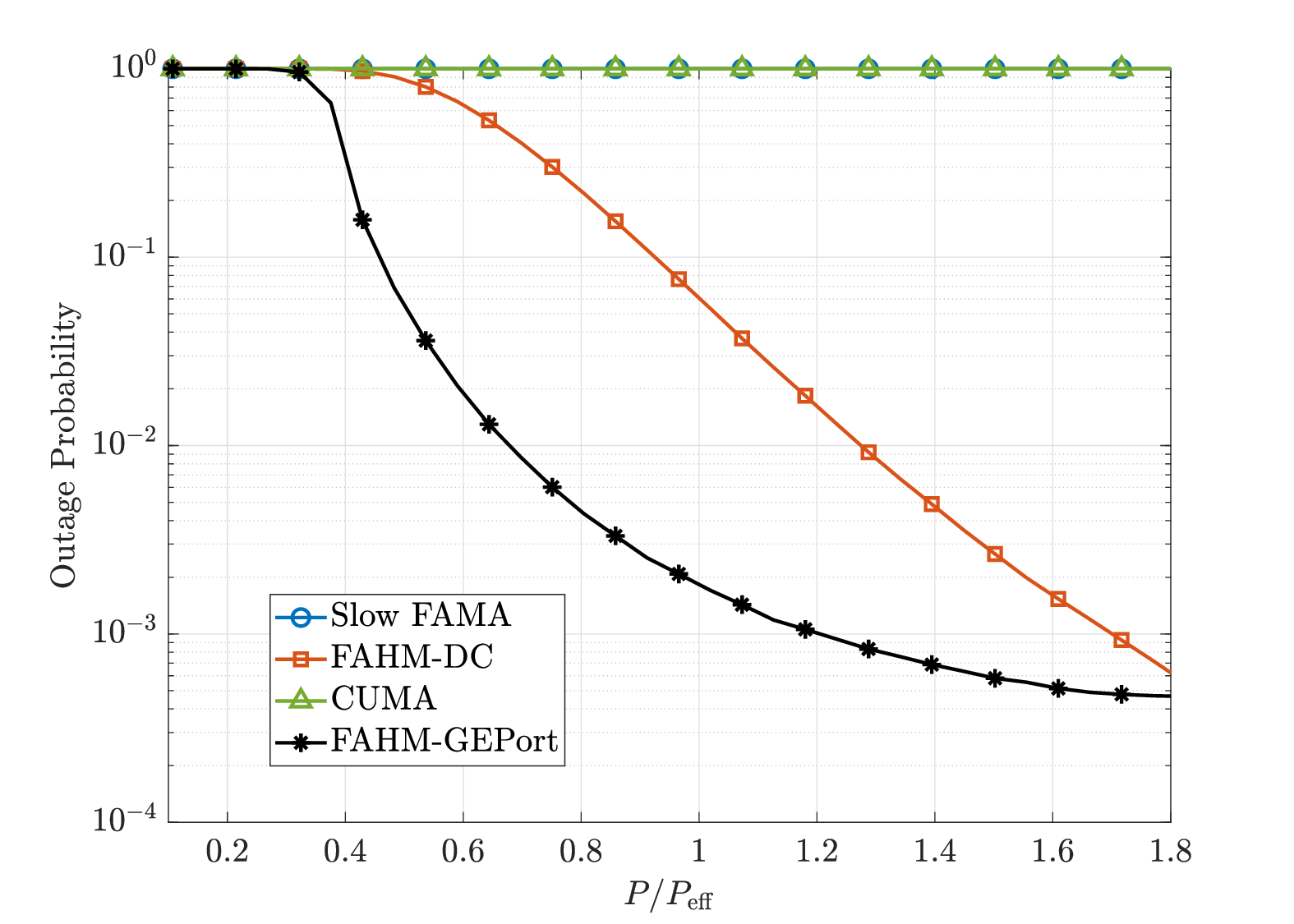}
\vspace{0mm}\caption{Outage probability for user $u$ versus the normalized selected-port ratio $P/P_{\mathrm{eff}}$ for the different schemes. Here, we assume a 2D FA case with $N= 7 \times 5$ and normalized aperture {$W= 2 \times 1$}. Also, $M=U=8$, $K=-10$ dB, $N_p=20$, $\mathrm{SNR}=15$ dB, and $\gamma=7$ bps/Hz.}
    \label{fig9}
\vspace{0mm}\end{figure}

In Fig.~\ref{fig6}, we evaluate the average SE per user achieved by the considered schemes as a function of the transmit SNR.  Specifically, Figs.~\ref{fig6}(a)--\ref{fig6}(c) correspond to different values of the number of selected ports $P$, whereas Fig.~\ref{fig6}(d) shows the elbow-based analysis used to identify the effective number of ports, denoted by $P_{\mathrm{eff}}$. In this setup, we consider a multiuser 2D FA configuration with $N=15\times 15=225$, and $M=U=70$, under Rayleigh fading channel. The normalized aperture is set to {$W = 3 \times 3$}, and the proposed FAHM-GEPort receiver is evaluated for three representative values of $P$. In particular, in Fig.~\ref{fig6}(c), the choice  {$P=85$} is not arbitrary, but is obtained from the effective-port indicator $P_{\mathrm{eff}}$, whose analytical expression is given by \eqref{Peff}. The curves in Figs.~\ref{fig6}(a)--\ref{fig6}(c) reveal a clear dependence of the achievable SE on the number of selected ports. When $P=30$ in Fig.~\ref{fig6}(a), the FAHM-GEPort scheme still outperforms the other benchmarks over the whole SNR range, but its SE remains almost flat and even shows a slight loss at high SNR, which indicates that the receiver is operating in interference-limited regime. In other words, the number of selected ports is too small to fully exploit the spatial diversity of the FA surface, and therefore increasing the transmit power cannot be effectively translated into SE gains. A similar behavior is observed for $P=60$ in Fig.~\ref{fig6}(b), where the proposed scheme improves with SNR but still tends to saturate at moderate-to-high SNR values. This suggests that, although selecting more ports alleviates the bottleneck, the receiver remains slightly under-dimensioned relative to the effective dimensionality of the channel. A markedly different trend appears when {$P=85$}, as shown in Fig.~\ref{fig6}(c). In this case, the FAHM-GEPort curve grows monotonically over the entire SNR range and no clear SE floor is observed. This is a particularly important result, since it indicates that once the number of selected ports reaches the effective number of ports required by the channel, the hybrid receiver is able to exploit the available spatial degrees of freedom to improve the SE. This interpretation is supported by the elbow analysis shown in Fig.~\ref{fig6}(d). The dominant-SINR curve, plotted as a function of the number of turned-off ports, presents a clear elbow around {$140$} switched-off ports. Since the surface of the FA contains $N=225$ ports, this yields { $P_{\mathrm{eff}} =N-140=85$}. Hence, the elbow criterion indicates that {$\approx85$} ports should remain active to preserve the dominant contribution of the channel. Figs.~\ref{fig6}(a)--\ref{fig6}(c) validate this assertion: only when {$P=P_{\mathrm{eff}}=85$} the proposed receiver avoids premature saturation and exhibits the desired SNR-dependent growth. This provides an attractive physical interpretation of $P_{\mathrm{eff}}$ as a practical design rule to select the minimum number of ports required to unlock the full benefit of FAHM.

In Fig.~\ref{fig8}, we assess the OP  by the considered schemes as a function of $\gamma$. Here, we consider a 2D FA configuration with $M=U=6$, $N = 5 \times 5$, and normalized aperture {$W = 2 \times 1$.} The Rice factor is fixed to $K=-10$ dB, the number of scattered paths is set to $N_p=5$, and the transmit SNR is fixed to $15$ dB. Moreover, two different port-selection settings are considered in this figure: $P=4$, represented by solid lines, and $P=P_{\mathrm{eff}}=14$, represented by dotted lines. Notice from all curves that, for the case $P=4$, FAHM-GEPort achieves the lowest OP over the whole threshold range, followed by FAHM-DC, while CUMA and slow-FAMA exhibit substantially worse reliability. A similar trend is observed for the case $P_{\mathrm{eff}}=14$, where the dotted curves show an improvement with respect to their solid-line counterparts, showcasing the performance gain associated to selecting a number of ports under the effective port criterion.

In Fig.~\ref{fig9}, we present the OP as a function of the normalized selected-port ratio $P/P_{\mathrm{eff}}$. In this scenario, we consider a multiuser 2D FA scenario with $M=U=8$, $N = 7 \times 5$, and normalized aperture {$W = 2 \times 1$.} The Rice factor is fixed to $K=-10$ dB, the number of scattered paths is set to $N_p=20$,  $\text{SNR}=15$ dB, and the outage threshold is fixed to $\gamma=7$ bps/Hz. From the curves, it is observed that the
normalized ratio also provides a clear interpretation of the different operating regions. When $P/P_{\mathrm{eff}}<1$, the receiver employs fewer ports than effectively required by the channel. In this regime, both FAHM-GEPort and FAHM-DC operate under a port-deficient condition, but FAHM-GEPort already exhibits a much sharper OP reduction than FAHM-DC. On the other hand, around $P/P_{\mathrm{eff}} \approx 1$, the FAHM-GEPort approaches the critical operating point where the number of selected ports becomes comparable to the effective channel dimensionality{\footnote{{Although $P_{\mathrm{eff}}$ is not the algebraic rank of the channel matrix, it can be interpreted as a port-domain effective dimensionality measure. Thus, $P/P_{\mathrm{eff}}\approx 1$ indicates that the selected ports are sufficient to capture the dominant effective channel subspace, while additional ports mainly provide redundant spatial information.}}}. This is precisely the region where the largest performance transition is observed. FAHM-GEPort experiences a very steep outage reduction and reaches a near-floor behavior around this normalized operating point. For $P/P_{\mathrm{eff}}>1$, the selected-port dimension exceeds the effective number of ports needed by the channel, and the improvement begins to saturate, as expected. 

{To investigate the effect of imperfect CSI, we model the estimated channel as $\widehat{\mathbf{H}}_u=\mathbf{H}_u+\mathbf{E}_u,$ where $\mathbf{E}_u$ represents the channel-estimation error. The entries of $\mathbf{E}_u$ are modeled as independent circularly symmetric complex Gaussian variables, 
with no spatial error correlation assumed. For each channel realization $\mathbf{h}_{u}$, the error variance is selected as $\sigma_e^2=\xi\frac{\|\mathbf{h}_{u}\|_2^2}{N},$ where $N$ is the number of candidate ports and $\xi$ controls the CSI accuracy. Hence, $\mathbb{E}\left[\frac{\|\mathbf{e}_{u}\|_2^2}{\|\mathbf{h}_{u}\|_2^2}\right]=\xi,$ 
and the corresponding CSI error level is expressed as $\mathrm{NMSE}({\mathrm{dB}})=10\log{10}(\xi)$. This model allows the sensitivity of the proposed receiver to different CSI accuracy levels to be evaluated without assuming a particular spatial covariance structure for the estimation error. We see in Fig. \ref{figCSIerror} that the performance of the FAHM and DC multiport schemes degrades as channel estimation error increases, but they still offer a substantial improvement in terms of SE for practical operating values of NMSE, especially in the range where the NMSE is comparable to the operating SNR.}
\begin{figure}[!t]
    \centering
\includegraphics[width=.98\columnwidth]{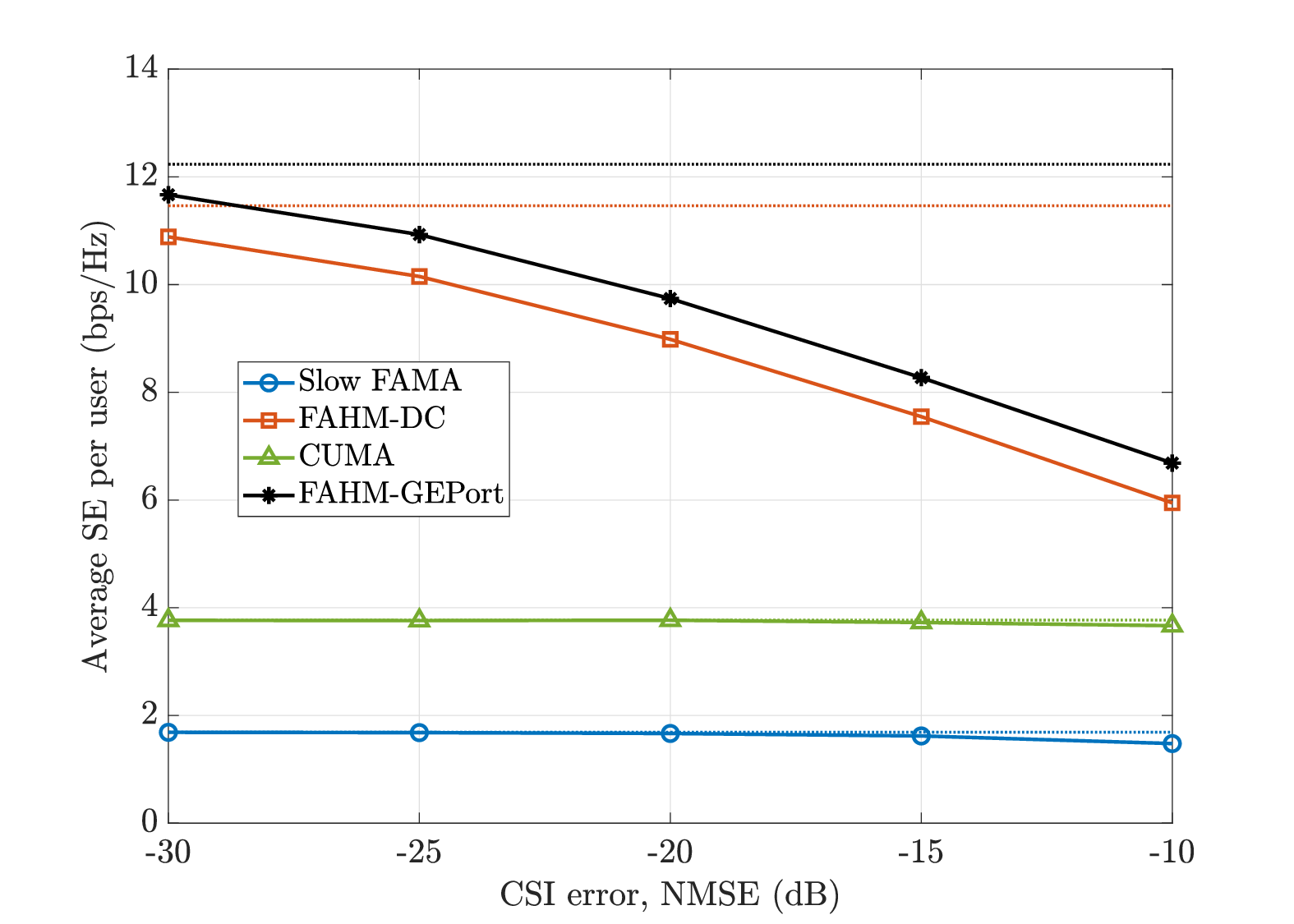}
\caption{{Average spectral efficiency per user versus CSI
estimation accuracy at ${\rm SNR}=20$ dB. All receiver strategies are designed
from the same estimated channel $\widehat{\mathbf{H}}_u=\mathbf{H}_u+\mathbf{E}_u$
and evaluated with the true channel $\mathbf{H}_u$. Dotted lines: perfect-CSI
benchmark. $U=6$, $N=15\times15$, $W_1\times W_2=5\times5$, $P=16$, $L=2$,
$K=-20$ dB, $N_p=80$.}}
    \label{figCSIerror}
\end{figure}

\begin{figure}[!t]
    \centering
\includegraphics[width=.98\columnwidth]{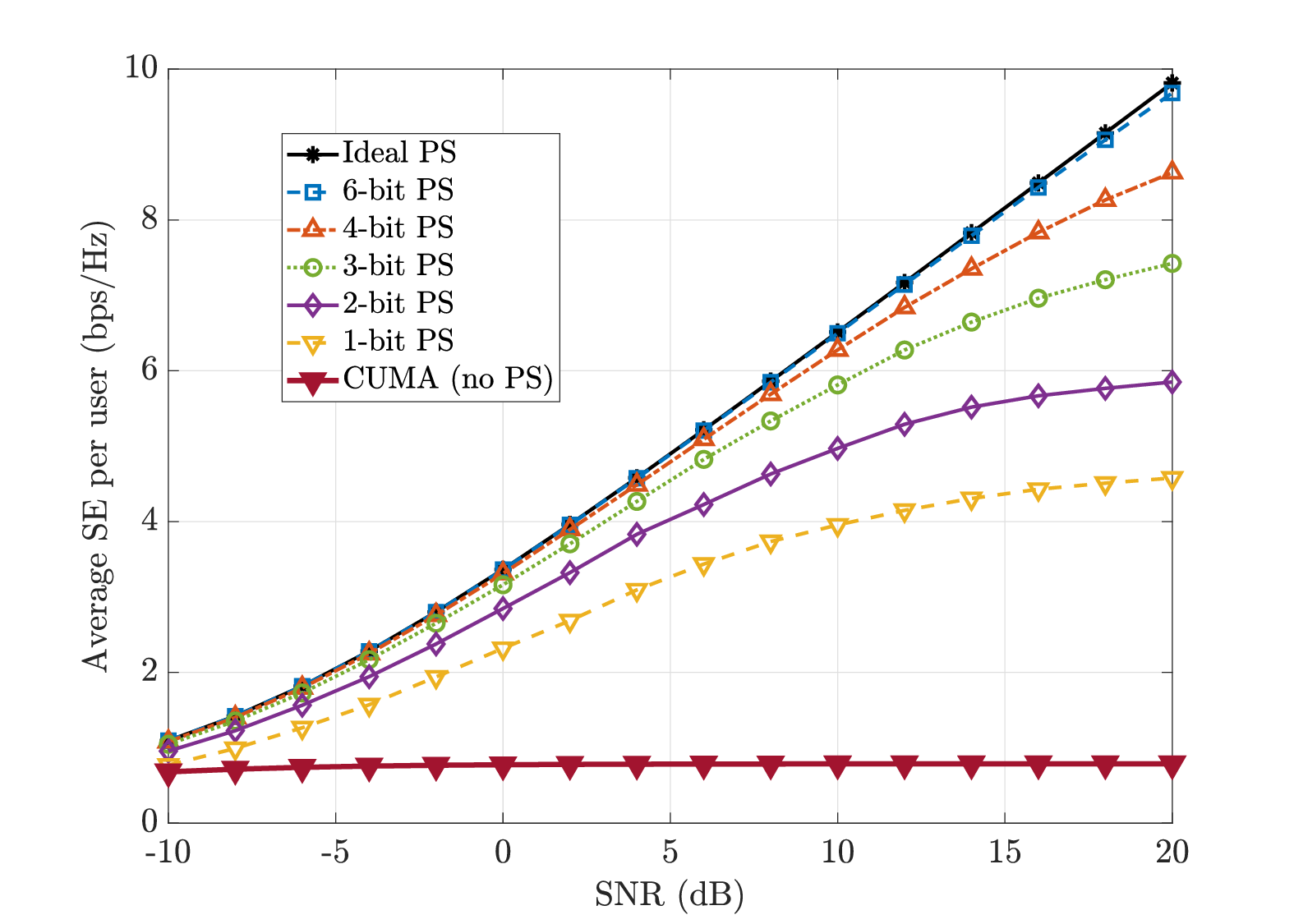}
\vspace{0mm}\caption{{Average SE per user versus SNR for FAHM with finite-resolution phase shifters ($L=2$ RF chains, $P=32$ selected ports), compared against CUMA, which requires no phase shifters. Here, we set $U=6$, $N=15\times15$, $W_1\times W_2=5\times5$, $K=10$ dB, and $N_p=80$.}}
    \label{figPS}
\vspace{0mm}\end{figure}
{To assess the robustness of FAHM to practical hardware imperfections, Fig. \ref{figPS} evaluates the impact of finite-resolution phase shifters on the achievable spectral efficiency. Here, we assume $U=6$ users, $N=15\times15$ ports over an aperture of $W_1\times W_2=5\times5$ wavelengths, $K=10$ dB, $N_p=80$ scattered paths, and $P=32$ selected ports combined through $L=2$ RF chains. The ideal analog phases are quantized to the closest available state for $1$, $2$, $3$, $4$, and $6$ bit phase shifters (PSs), while the digital $\mathbf{w}_u$ is adapted accordingly. The results show that FAHM is highly robust to phase quantization. The $6$-bit implementation is virtually indistinguishable from the ideal-PS case over the whole SNR range, with a loss below $1.4\%$, while $4$-bit phase shifters retain performance close to the ideal solution, with only a moderate degradation at high SNR. Even with $3$-bit phase shifters, the proposed architecture preserves a substantial fraction of the ideal combining gain. We also see that FAHM substantially outperforms CUMA (which requires no phase shifters) over the whole SNR range even with $1$-bit phase shifters.} 

\begin{figure}[!t]
    \centering
    \includegraphics[width=.98\columnwidth]{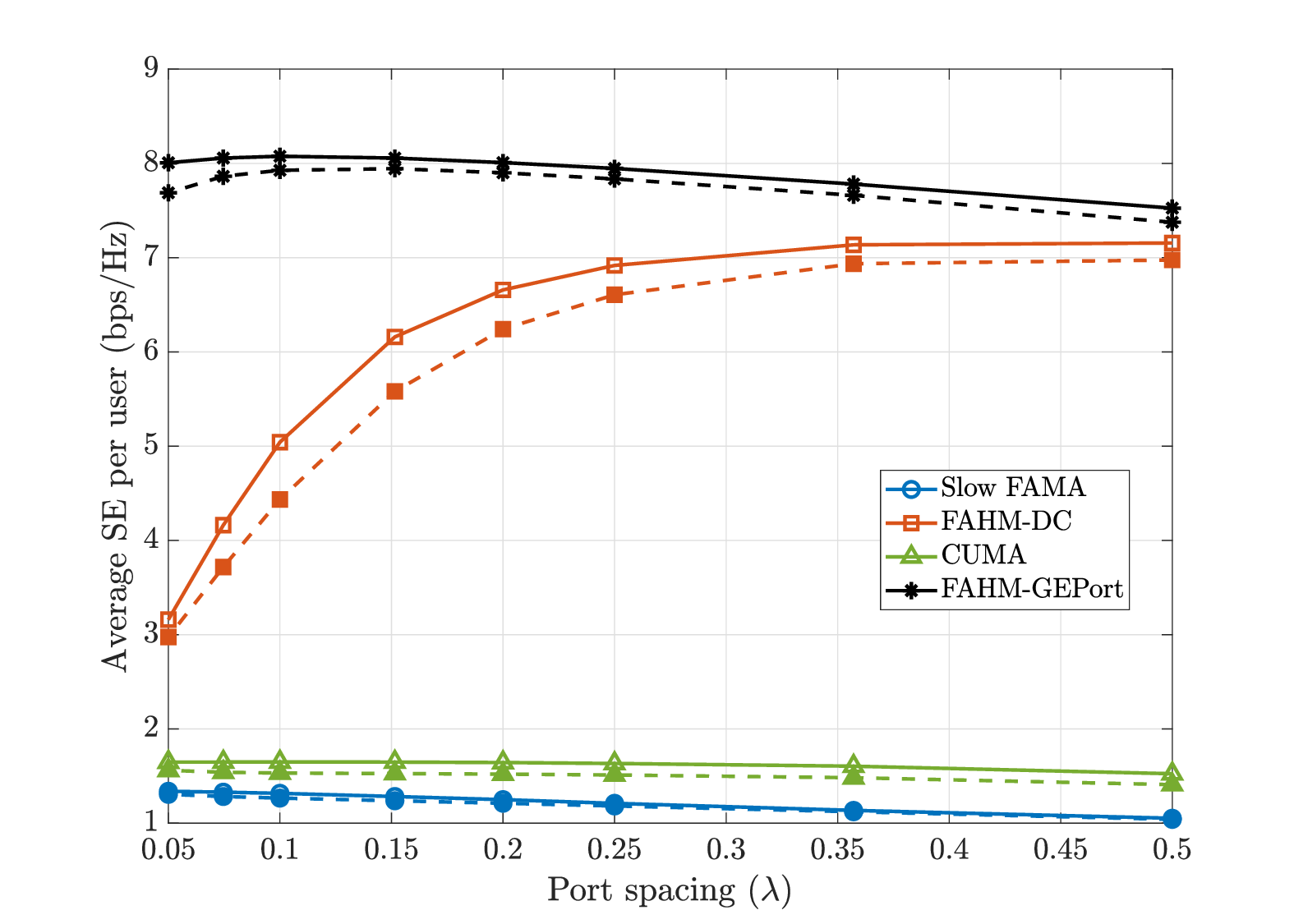}
    \vspace{0mm}\caption{{Average SE per user versus port spacing with mutual coupling. The FAS aperture is
    $W=5\lambda$, with  $d\in[0.05\lambda,0.5\lambda]$. Simulation parameters are $U=5$, $P=8, L=2$, $K=-20$ dB, $N_p=30$, and $\mathrm{SNR}=15$ dB. Solid/dashed curves denote cases without/with mutual coupling, respectively.}}
    \label{fig:MC}
\vspace{0mm}\end{figure}
{To complement the ideal-hardware analysis, we have incorporated mutual coupling between the FAS ports in Fig.~\ref{fig:MC}. Following \cite{Dinis26}, the coupled received signal is characterized by $\mathbf{\Gamma}=Z_T(\mathbf{Z}+Z_T\mathbf{I})^{-1}$, where $\mathbf{Z}$ is the antenna impedance matrix and $Z_T=50~\Omega$. The ports are modeled as half-wavelength dipoles, consistent with \cite{Dinis26}, and the resulting coupling matrix is applied to the channel as $\mathbf{H}^{\rm eq}_u = \mathbf{\Gamma}^{\rm rx}_u\mathbf{H}_u$. We have considered $W_{\mathrm{FAS}}=5\lambda$, varying the candidate-port spacing from $0.5\lambda$ down to $0.05\lambda$.  In general, the effect of coupling becomes more noticeable for very closely spaced ports. Importantly, FAHM-GEPort remains the best-performing scheme over the complete spacing range, both with and without mutual coupling. When coupling is considered, 
we see that the FAHM-GEPort scheme exhibits a good resilience to denser port spacing. While coupling does not fundamentally alter the number of dominant eigenvalues of the channel spatial covariance matrix, it makes their decay faster \cite{Ramirez2025}; hence, $P=8$ is enough to capture the reduced dominant eigen-subspace and to filter the effect of $U-1=4$ interferers. Note that the degradation of the FAHM-DC scheme is comparatively larger, due to its lack of flexibility to adapt its port selection mechanism to
the spatial eigenspectrum of the coupled channel. In all instances, performance of FAHM-based alternatives is consistently better than those of CUMA and single-port slow-FAMA even when coupling is assumed.}

\begin{table}[!t]
\centering
\vspace{0mm}\caption{Computational Complexity and Average Elapsed Time}
\renewcommand{\arraystretch}{0.93}
\label{tab:complexity}
\begin{tabular}{|l|c|c|}
\hline
\textbf{Method} & \textbf{Complexity} & \textbf{Time (ms)} \\
\hline
Slow-FAMA~\cite{SlowFAMA} & $\mathcal{O}(NU)$ & {0.77}\\
DC~\cite{GoLo26} & {$\mathcal{O}(NU + L^3)$} & {0.59} \\
CUMA~\cite{CUMA} & $\mathcal{O}(N)$ & {4.51} \\
GEPort~\cite{GoLo26} & $\mathcal{O}((N-P)N^3)$ & {6903.50} \\
GEPort + Rank-1 (prop.) & $\mathcal{O}((N-P)N^2)$ & {74.61} \\
\hline
\end{tabular}
\vspace{0mm}
\label{tabla1}
\end{table}

Finally, Table~\ref{tab:complexity} summarizes\footnote{{The computational cost of all considered methods was reevaluated in MATLAB R2022b on a Linux DGX Station server equipped with an Intel Xeon processor and an NVIDIA Tesla V100-DGXS-32GB GPU. All reported timings were obtained in CPU-only mode.}}  the computational complexity and the average elapsed time of the considered schemes. {Channel estimation is not included since it is common to all of them. To ensure a consistent runtime comparison, we calculated all execution times using the same set of $20$ independent channel realizations at $\mathrm{SNR}=25$ dB. Following the large-system configuration considered in Fig.~\ref{fig6}(c), we set $N=225$ and $M=U=70$. For the conventional DC benchmark, we set $P_{\mathrm{DC}}=2$ selected ports and $L=2$ RF chains; hence, $P_{\mathrm{DC}}=L$, and its complexity is expressed as $\mathcal{O}(NU+L^3)$. For GEPort and its Rank-1 implementation, the retained-port dimension is $P_{\mathrm{GEPort}}=P_{\mathrm{eff}}=85$.} {Under the adopted complexity model, CUMA retains the favorable asymptotic scaling $\mathcal{O}(N)$ as $U$ grows. Nevertheless, this does not make it universally preferable: for $U=70$, its measured runtime is $4.51$ ms, compared with $0.77$ ms for slow-FAMA and $0.59$ ms for DC, while the numerical results in Fig.~\ref{fig6}(c) show that CUMA attains a lower SE than FAHM-GEPort. Thus, CUMA is attractive when low asymptotic complexity is the primary design objective, whereas FAHM-GEPort targets a different performance--complexity tradeoff.} As can be observed, classical GEPort scheme as proposed in \cite{GoLo26} entails a larger computational cost, both in terms of asymptotic complexity and measured runtime, due to the repeated matrix operations involved in the generalized eigenvector-based selection process. {In addition, using the hybrid decomposition method from \cite{ZhMoKu05} over the resulting combining vector presents linear complexity, which is negligible compared to that of GEPort.} Nevertheless, GEPort delivers substantial performance gains, which makes it a suitable high-performance benchmark. {The proposed GEPort + Rank-1 implementation significantly alleviates this burden by reducing the complexity from $\mathcal{O}((N-P)N^3)$ to $\mathcal{O}((N-P)N^2)$. Consistently, the average elapsed time decreases from $6903.50$ ms to $74.61$ ms, corresponding to a reduction of approximately $98.9\%$.}

\vspace{0mm}
\section{Conclusion}
 \label{conclusiones}
In this work, a novel hybrid multiport receiver for slow fluid antenna multiple access was proposed. The resulting FAHM architecture was shown to preserve the main benefits of multiport reception while requiring only a small number of RF chains. By introducing an equivalent combiner over the selected-port subspace, the receiver design was reformulated in a hybrid analog-digital form, for which only two RF chains are sufficient in the considered single-stream slow-FAMA setting. To formally determine the dimension of the selected-port set, an effective-port criterion was introduced to identify the minimum number of selected ports required to preserve the dominant channel contribution before hybrid combining. Then, an efficient version of the GEPort method was adopted, using a reduced-complexity implementation that leverages the rank-1 structure of the useful-signal matrix. The proposed FAHM formulation provides a unifying framework for multiport reception, with CUMA arising as a particular structured instance. Numerical results showed that FAHM-GEPort consistently outperforms the considered benchmarks in terms of spectral efficiency and outage probability, while $P_{\mathrm{eff}}$ was shown to be an accurate design criterion to achieve hybrid multiport gains.

{While the present work established the foundation of hybrid schemes for slow FAMA, several aspects that deserve further investigation are listed as follows: on the one hand, the incorporation of spatial noise correlation induced by mutual coupling at ultra-dense port spacings will require adapting the underlying signal model of GEPort to non-white noise environments. Similarly, accounting for additional HW imperfections such as power attenuation, insertion losses, and signal routing overheads intrinsic to high-dimensional analog combining networks requires further attention. Furthermore, addressing the operational challenges related to imperfect CSI through robust channel estimation algorithms tailored to the two-RF chain framework will be essential for practical deployment. Finally, validating the proposed architecture against specific physical FA embodiments by integrating full-wave electromagnetic impedance matrices and realistic reconfigurable architectures will provide a complete characterization of the trade-offs between physical array geometry, circuit-level losses, and achievable spatial efficiency.}
\vspace{0mm}
\appendices
\section{GEPort Complexity Reduction} \label{apA}
The Generalized Eigenvector Port Selection method proposed in \cite{GoLo26} has a computational cost that is cubic in the number of selected ports $\mathcal{O}(N^3)$, and its iterative nature leads to overall computational times that are dominated by the costs of each step. To alleviate this computational burden, we first note that the approach in \cite{GoLo26} computes the generalized eigen-decomposition of the matrix pair $(\mathbf{A}, \mathbf{B})$. Now, since $\mathbf{A}$ is a rank-$1$ matrix by its definition, the dominant eigenvector can be obtained without the explicit decomposition. For some user $u$, the matrix is defined as $\mathbf A_u=\mathbf{H}_u\mathbf{p}_u\mathbf{p}_u^H\mathbf{H}_u^H$. In addition, any generalized eigenvalue $\mathbf v$ satisfies the equation $\mathbf A_u\mathbf v=\lambda\mathbf{B}_u\mathbf v$. By introducing the equivalent channel $\mathbf{h}_u=\mathbf{H}_u\mathbf{p}_u$ and substituting in the latter expression, we get that
\begin{equation}
        \mathbf{h}_u\mathbf{h}_u^H\mathbf v=\lambda\mathbf{B}_u\mathbf v.
        \label{eq:eigenRank1}
\end{equation}
In the following, we remove the subindex $u$ for ease of exposition. Now, since $\mathbf{B}$ is a regular matrix, we can rearrange the terms taking advantage of the definition of $\mathbf A$, so that 
        $\mathbf{B}^{-1}\mathbf{h}\mu=\mathbf v$,
where the scalar $\mu=\lambda^{-1}\mathbf{h}^H\mathbf v$ is unknown, but it does not affect the equality in  \eqref{eq:eigenRank1}. Then, the dominant eigenvalue can be obtained without performing the decomposition as $\mathbf v=\mathbf{B}^{-1}\mathbf{h}$. Still, the computation of the inverse of $\mathbf{B}^{-1}$ is costly and has to be recomputed at each iteration, due to the removal of the row and the column from $\mathbf{B}$ associated with the discarded port. We can further improve the efficiency of the method by dealing with this matrix operation. First, notice that the permutation matrix $\mathbf P_{(i)}$, that exchanges the position of the $i$-th row and the first row of a matrix, is orthogonal, $\mathbf{P}^{-1}_{(i)} = \mathbf{P}_{(i)}^T$, and therefore
\begin{align}
\mathbf{B}^{-1}_{(i)}
=(\mathbf{P}_{(i)} \mathbf{B} \mathbf{P}^T_{(i)})^{-1}
=\mathbf{P}_{(i)} \mathbf{B}^{-1} \mathbf{P}^T_{(i)},
\end{align}
where we denote as $\mathbf{B}^{-1}_{(i)}$ the matrix which results from permuting the $i$-th row and the $i$-th column with the first row and column of $\mathbf{B}^{-1}$, respectively, so that the permutation and the inverse commute. Hence, we can focus on dropping the first row and column of $\mathbf{B}^{-1}$ without loss of generality. Let us now introduce the block decomposition of $\mathbf B$ as
\begin{equation}
\mathbf{B} =
\begin{pmatrix}
b_{11} & \mathbf{b}^T \\
\mathbf{b} & \tilde{\mathbf{B}}
\end{pmatrix},
\end{equation}
and let $\nu = b_{11} - \mathbf{b}^T \tilde{\mathbf{B}}^{-1} \mathbf{b}$
denote the Schur complement of $\tilde{\mathbf{B}}$ in $\mathbf{B}$. Using the block matrix inversion formula, one obtains
\begin{align}
\mathbf{B}^{-1}
&=\begin{pmatrix}
\nu^{-1} & -\mathbf{c}^T \\
-\mathbf{c} &
\mathbf{C}
\end{pmatrix},
\end{align}
where $\mathbf c=\tilde{\mathbf{B}}^{-1}\mathbf{b}\nu^{-1}$ and $\mathbf{C}=\tilde{\mathbf{B}}^{-1}+\tilde{\mathbf{B}}^{-1}\mathbf{b}\nu^{-1}\mathbf{b}^T\tilde{\mathbf{B}}^{-1}$
Finally, the desired inverse after removing the first row and column of $\mathbf{B}$ is readily obtained as
\begin{equation}
\tilde{\mathbf{B}}^{-1} =\mathbf{C}-\nu\mathbf{c}\mathbf{c}^T.
\end{equation}
By using the previous expression, the computational cost of successive iterations is reduced by one order of magnitude, yielding a complexity of $\mathcal{O}(N^2)$ per iteration.

\vspace{0mm}
\section{Effective-Port Criterion}
\label{apB}
Let $\mathbf{v}\in\mathbb{C}^{N}$ denote the dominant generalized eigenvector associated with the current GEPort iteration, where each entry $v_i$ represents the complex weight assigned to the $i$-th port after combining. The squared magnitude $|v_i|^2$ can be interpreted as the contribution (or energy) of port $i$ to the overall combined signal. We define the normalized per-port energy weights as
\begin{equation}
\label{def:pi}
    \pi_i = \frac{|v_i|^2}{\sum_{j=1}^{N}|v_j|^2}, \qquad i=1,\ldots,N,
\end{equation}
so that $\pi_i \ge 0$ and $\sum_{i=1}^{N}\pi_i = 1$. The effective number of selected ports is defined as the inverse participation ratio of the energy distribution $\{\pi_i\}_{i=1}^{N}$, namely,
\begin{equation}
    P_{\mathrm{eff}} = \frac{1}{\sum_{i=1}^{N}\pi_i^2}.
\end{equation}
Substituting \eqref{def:pi} into the above expression and 
after straightforward simplification, we obtain:
\begin{equation}
    P_{\mathrm{eff}}
=\frac{\left(\sum_{i=1}^{N}|v_i|^2\right)^2}
         {\sum_{i=1}^{N}|v_i|^4}.
\end{equation}
If $\mathbf{v}$ is normalized so that $\|\mathbf{v}\|_2=1$, then $\sum_{i=1}^{N}|v_i|^2 = 1$ and the previous expression reduces to
\begin{equation}
    P_{\mathrm{eff}} = \frac{1}{\sum_{i=1}^{N}|v_i|^4}.
\end{equation}

\vspace{0mm}
\section{CUMA as a special case of FAHM}
\label{apC}
In this section, we show how the {reference CUMA} method in \cite{CUMA} can be seen as a particular case of the proposed structure. Recall that the detected symbol $\hat{z}_u = \mathbf{w}_u^H \mathbf{F}_u^H \mathbf{S}_u^T \mathbf{x}_u$ is obtained with a selection matrix $\mathbf{S}_u$ and the matrix $\mathbf{F}_u$ representing the analog network of phase-shifter. To mimic the CUMA strategy, we first set the selection matrix $\mathbf{S}_u$ including all the ports for the in-phase and quadrature components of the channel, denoted as $\mathcal{P}_I$ and $\mathcal{P}_Q$. Next, we employ a matrix $\mathbf{F}_u\in\mathbb{C}^{P\times 4}$ defined as follows 
\begin{align}
     [\mathbf{F}_u]_{p,1}=\begin{cases}e^{j\phi_p}\cos(\phi_p) & \text{if $p \in \mathcal{P}_I$} \\
    0 & \text{otherwise}
    \end{cases}\\
    [\mathbf{F}_u]_{p,2}=\begin{cases}e^{j\phi_p}\sin(\phi_p) & \text{if $p \in \mathcal{P}_I$} \\
    0 & \text{otherwise}
    \end{cases}
\end{align}
where $[\mathbf{x}_u]_p=\rho_pe^{j\phi_p}$. This strategy is replicated for rows $3$ and $4$ of $\mathbf{F}_u$, associated with the ports selected for the quadrature component of the channel, $\mathcal{P}_Q$. Similar to \cite{CUMA}, we define the aggregated signal as 
\begin{equation}
\mathbf r_u=\mathbf{F}_u^H\mathbf S_u^T\mathbf{x}_u
\end{equation}
where $\mathbf r_u=[r_{u,1}^{\mathrm{I}},
r_{u,2}^{\mathrm{I}},
r_{u,1}^{\mathrm{Q}},
r_{u,2}^{\mathrm{Q}}]^T$. To recover the two components, CUMA solves a system of equations. This can be replicated under our approach by calculating $\hat{z}_u = \mathbf{w}_u^H \mathbf{r}_u$,
where the vector $\mathbf{w}_u^H$ is given by
\begin{equation}
    \mathbf{w}_u^H=\begin{bmatrix}1 & j\end{bmatrix}\begin{bmatrix}
\sum_{p \in \mathcal{P}_I} \Re\!\left([\mathbf{h}_u]_p\right) &
-\sum_{p \in \mathcal{P}_I} \Im\!\left([\mathbf{h}_u]_p\right) \\
\sum_{p \in \mathcal{P}_I} \Im\!\left([\mathbf{h}_u]_p\right) &
\sum_{p \in \mathcal{P}_I} \Re\!\left([\mathbf{h}_u]_p\right) \\
\sum_{p \in \mathcal{P}_Q} \Re\!\left([\mathbf{h}_u]_p\right) &
-\sum_{p \in \mathcal{P}_Q} \Im\!\left([\mathbf{h}_u]_p\right) \\
\sum_{p \in \mathcal{P}_Q} \Im\!\left([\mathbf{h}_u]_p\right) &
\sum_{p \in \mathcal{P}_Q} \Re\!\left([\mathbf{h}_u]_p\right)
\end{bmatrix}^{\dagger},
\label{eqWcuma}
\end{equation}
$\mathbf h_u = \mathbf{H}_u \mathbf{p}_u$.

\bibliographystyle{IEEEtran}

\bibliography{References}

@ARTICLE{Clerckx2024MA6G,
  author={Clerckx, Bruno and Mao, Yijie and Yang, Zhaohui and Chen, Mingzhe and Alkhateeb, Ahmed and Liu, Liang and Qiu, Min and Yuan, Jinhong and Wong, Vincent W. S. and Montojo, Juan},
  journal={Proc. IEEE}, 
  title={{Multiple Access Techniques for Intelligent and Multifunctional 6G: Tutorial, Survey, and Outlook}}, 
  year={2024},
  volume={112},
  number={7},
  pages={832-879},
  doi={10.1109/JPROC.2024.3409428}}

@ARTICLE{Wang2024XLMIMO,
  author={Wang, Zhe and Zhang, Jiayi and Du, Hongyang and Sha, Wei E. I. and Ai, Bo and Niyato, Dusit and Debbah, Mérouane},
  journal={IEEE Wireless Commun.}, 
  title={{Extremely Large-Scale MIMO: Fundamentals, Challenges, Solutions, and Future Directions}}, 
  year={2024},
  volume={31},
  number={3},
  pages={117-124},
  doi={10.1109/MWC.132.2200443}}

@article{BruceLee2020,
author = {Wong, Kai-Kit and Tong, Kin-Fai and Zhang, Yangyang and Zhongbin, Zheng},
title = {{Fluid Antenna System for 6G: When Bruce Lee Inspires Wireless Communications}},
journal = {Electron. Lett.},
volume = {56},
number = {24},
pages = {1288-1290},
year = {2020}
}

@ARTICLE{Lu2025FluidAntennas,
  author={Lu, Wen-Jun and He, Chun-Xing and Zhu, Yongxu and Tong, Kin-Fai and Wong, Kai-Kit and Shin, Hyundong and Cui, Tie-Jun},
  journal={IEEE Commun. Mag.},
  title={Fluid Antennas: Reshaping Intrinsic Properties for Flexible Radiation Characteristics in Intelligent Wireless Networks},
  year={2025},
  volume={63},
  number={5},
  pages={40-45},
  doi={10.1109/MCOM.002.2400490}
}

@article{Fang2025LiquidMetal,
  author  = {\textcolor{blue}{Fang, Xiaochuan and Qu, Zhishu and Alqurashi, Khaled Yahya and Kelly, James R. and Alkaraki, Shaker M. M.}},
  title   = {{\textcolor{blue}{A Liquid Metal-Based Frequency Reconfigurable Patch Antenna With Cross Placed Fluidic Channel}}},
  journal = {\textcolor{blue}{Microw. Opt. Technol. Lett.}},
  year    = {\textcolor{blue}{2025}},
  volume  = {\textcolor{blue}{67}},
  pages   = {\textcolor{blue}{e70112}},
  doi     = {\textcolor{blue}{10.1002/mop.70112}}
}

@ARTICLE{Hong2026Survey,
  author={Hong, Hanjiang and Wong, Kai-Kit and Chae, Chan-Byoung and Xu, Hao and Guo, Xinghao and Rostami Ghadi, Farshad and Chen, Yu and Xu, Yin and Liu, Baiyang and Tong, Kin-Fai and Zhang, Yangyang},
  journal={IEEE Trans. Netw. Sci. Eng.},
  title={A Contemporary Survey on Fluid Antenna Systems: Fundamentals and Networking Perspectives},
  year={2026},
  volume={13},
  pages={2305-2328},
  doi={10.1109/TNSE.2025.3613225}
}

@ARTICLE{FastFAMA,
  author={Wong, Kai-Kit and Tong, Kin-Fai and Chen, Yu and Zhang, Yangyang},
  journal={IEEE Commun. Lett.},
  title={Fast Fluid Antenna Multiple Access Enabling Massive Connectivity},
  year={2023},
  volume={27},
  number={2},
  pages={711-715},
  doi={10.1109/LCOMM.2022.3222574}
}

@ARTICLE{Xu2024Outage,
  author={Xu, Hao and Wong, Kai-Kit and New, Wee Kiat and Tong, Kin-Fai and Zhang, Yangyang and Chae, Chan-Byoung},
  journal={IEEE Trans. Wireless Commun.},
  title={Revisiting Outage Probability Analysis for Two-User Fluid Antenna Multiple Access System},
  year={2024},
  volume={23},
  number={8},
  pages={9534-9548},
  doi={10.1109/TWC.2024.3363499}
}

@ARTICLE{Rappaport2015Wideband,
  author={Rappaport, Theodore S. and MacCartney, George R. and Samimi, Mathew K. and Sun, Shu},
  journal={IEEE Trans. Commun.}, 
  title={Wideband Millimeter-Wave Propagation Measurements and Channel Models for Future Wireless Communication System Design}, 
  year={2015},
  volume={63},
  number={9},
  pages={3029-3056},
  doi={10.1109/TCOMM.2015.2434384}}

@ARTICLE{Hong2025BlockCodedFAMA,
  author={Hong, Hanjiang and Wong, Kai-Kit and Tong, Kin-Fai and Xu, Hao and Li, Haoyang},
  journal={Electron. Lett.},
  title={{5G-Coded Fluid Antenna Multiple Access over Block Fading Channels}},
  year={2025},
  volume={61},
  number={1},
  pages={e70166},
  doi={10.1049/ell2.70166}
}

@ARTICLE{Hong2025FastCodedFAMA,
  author={Hong, Hanjiang and Wong, Kai-Kit and Tong, Kin-Fai and Shin, Hyundong and Zhang, Yangyang},
  journal={IEEE Wireless Commun. Lett.},
  title={Coded Fluid Antenna Multiple Access over Fast Fading Channels},
  year={2025},
  volume={14},
  number={4},
  pages={1249-1253},
  doi={10.1109/LWC.2025.3540668}
}

@ARTICLE{Waqar2025TurboFAMA,
  author={Waqar, Noor and Wong, Kai-Kit and Chae, Chan-Byoung and Murch, Ross},
  journal={IEEE Trans. Wireless Commun.},
  title={Turbocharging Fluid Antenna Multiple Access},
  year={2025},
  doi={10.1109/TWC.2025.3607824},
  note={Early Access}
}

@ARTICLE{CUMA,
  author={Wong, Kai-Kit and Chae, Chan-Byoung and Tong, Kin-Fai},
  journal={IEEE Trans. Wireless Commun.}, 
  title={Compact Ultra Massive Antenna Array: A Simple Open-Loop Massive Connectivity Scheme}, 
  year={2024},
  volume={23},
  number={6},
  pages={6279-6294},
  doi={10.1109/TWC.2023.3330932}}

@ARTICLE{SlowFAMA,
  author={Wong, Kai-Kit and Morales-Jimenez, David and Tong, Kin-Fai and Chae, Chan-Byoung},
  journal={IEEE Trans. Commun.}, 
  title={{Slow Fluid Antenna Multiple Access}}, 
  year={2023},
  volume={71},
  number={5},
  pages={2831-2846},
  doi={10.1109/TCOMM.2023.3255904}}

@ARTICLE{FAMA,
  author={Wong, Kai-Kit and Tong, Kin-Fai},
  journal={IEEE Trans. Wireless Commun.}, 
  title={{Fluid Antenna Multiple Access}}, 
  year={2022},
  volume={21},
  number={7},
  pages={4801-4815},
  doi={10.1109/TWC.2021.3133410}}

@ARTICLE{Wong2021,
  author={Wong, Kai-Kit and Shojaeifard, Arman and Tong, Kin-Fai and Zhang, Yangyang},
  journal={IEEE Trans. Wireless Commun.}, 
  title={Fluid Antenna Systems}, 
  year={2021},
  volume={20},
  number={3},
  pages={1950-1962},
  doi={10.1109/TWC.2020.3037595}}

@INPROCEEDINGS{CUMA2,
  author={Wong, Kai-Kit},
  booktitle={Proc. IEEE Workshop Signal Process. Adv. Wirel. Commun. (SPAWC)}, 
  title={{Compact Ultra Massive Array (CUMA) with 4 RF Chains for Massive Connectivity}}, 
  year={2024},
  volume={},
  number={},
  pages={286-290},
  doi={10.1109/SPAWC60668.2024.10694593}}

@Article{ZhMoKu05,
  author  = {Xinying Zhang and A. F. Molisch and Sun-Yuan Kung},
  title   = {{Variable-phase-shift-based RF-baseband codesign for MIMO antenna selection}},
  journal = {IEEE Trans. Signal Process.},
  year    = {2005},
  volume  = {53},
  number  = {11},
  pages   = {4091-4103},
  month   = {November},
  issn    = {1053-587X},
  doi     = {10.1109/TSP.2005.857024},
}

@Article{GoLo26,
  author  = {González-Coma, José P. and López-Martínez, F. Javier},
  journal = {IEEE Wireless Commun. Lett.},
  title   = {{Slow Fluid Antenna Multiple Access With Multiport Receivers}},
  year={2026},
  volume={15},
  number={},
  pages={1280-1284},  doi     = {10.1109/LWC.2026.3651420},
}

@article{hong2025multiport,
  title = {{Multi-Port Selection for FAMA: Massive Connectivity with Fewer RF Chains than Users}},
  author = {Hong, Hanjiang and Wong, Kai-Kit and Zhu, Xusheng and Xu, Hao and Xiao, Han and Rostami Ghadi, Farshad and Shin, Hyundong},
  journal = {arXiv preprint arXiv:2511.17897},
  year = {2025},
  eprint = {2511.17897},
  archivePrefix = {arXiv},
  primaryClass = {cs.IT},
  doi = {10.48550/arXiv.2511.17897},
  }

@article{perezadan2026greedy,
  title = {Greedy and Transformer-Based Multi-Port Selection for Slow Fluid Antenna Multiple Access},
  author = {Perez-Adan, Darian and Gonzalez-Coma, Jose P. and L\'opez-Mart\'inez, F. Javier and Castedo, Luis},
  journal = {arXiv preprint arXiv:2604.04589},
  year = {2026},
  eprint = {2604.04589},
  archivePrefix = {arXiv},
  primaryClass = {cs.AI},
  doi = {10.48550/arXiv.2604.04589},
 }

@ARTICLE{waqar2023dlslow,
  author={Waqar, Noor and Wong, Kai-Kit and Tong, Kin-Fai and Sharples, Adrian and Zhang, Yangyang},
  journal={IEEE Commun. Lett.}, 
  title={Deep Learning Enabled Slow Fluid Antenna Multiple Access}, 
  year={2023},
  volume={27},
  number={3},
  pages={861-865},
  doi={10.1109/LCOMM.2023.3237595}}

@ARTICLE{Hong25a,
  author={Hong, Hanjiang and Wong, Kai-Kit and Xu, Hao and Xu, Yin and Shin, Hyundong and Murch, Ross and He, Dazhi and Zhang, Wenjun},
  journal={IEEE Trans. Wireless Commun.}, 
  title={{Downlink OFDM-FAMA in 5G-NR Systems}}, 
  year={2025},
  volume={24},
  number={12},
  pages={10116-10132},
  doi={10.1109/TWC.2025.3577771}}

@ARTICLE{Nossek10,
  author={Ivrlač, Michel T. and Nossek, Josef A.},
  journal={IEEE Trans. Circuits Syst. I Regul. Pap.}, 
  title={{Toward a Circuit Theory of Communication}}, 
  year={2010},
  volume={57},
  number={7},
  pages={1663-1683},
  doi={10.1109/TCSI.2010.2043994}}

@ARTICLE{Zhang2024,
  author={\textcolor{blue}{Zhang, Jichen \emph{et. al.}}},
  journal={\textcolor{blue}{IEEE Open J. Antennas Propag.}},
  title={{\textcolor{blue}{A Novel Pixel-Based Reconfigurable Antenna Applied in Fluid Antenna Systems With High Switching Speed}}},
  year={\textcolor{blue}{2025}},
  volume={\textcolor{blue}{6}},
  number={\textcolor{blue}{1}},
  pages={\textcolor{blue}{212-228}},
  doi={\textcolor{blue}{10.1109/OJAP.2024.3489215}}}

@article{Zhang2025,
      title={Scalable Pixel-based Reconfigurable Beamforming Networks for Designing Fluid Antenna Systems}, 
      author={Jichen Zhang and Junhui Rao and Tianqu Kang and Zhaoyang Ming and Yijun Chen and Alikhan Umirbayev and Chi-Yuk Chiu and Ross Murch},
      year={2025},
      journal={arXiv preprints, arXiv:2512.03703 [eess.SP]}, 
      }

@ARTICLE{Zhu2024,
  author={Zhu, Lipeng and Ma, Wenyan and Ning, Boyu and Zhang, Rui},
  journal={IEEE Trans. Wireless Commun.}, 
  title={Movable-Antenna Enhanced Multiuser Communication via Antenna Position Optimization}, 
  year={2024},
  volume={23},
  number={7},
  pages={7214-7229},
  doi={10.1109/TWC.2023.3338626}}

@article{Ramirez2025,
      title={Metasurface-based Fluid Antennas: from Electromagnetics to Communications Model}, 
      author={Pablo Ramírez-Espinosa and Cleofás Segura-Gómez and Ángel Palomares-Caballero and F. Javier López-Martínez and David Morales-Jiménez},
      year={2025},
      journal={arXiv e-prints, arXiv:2507.17982 [eess.SP]},}

@article{Shen2024,
      title={Design and Implementation of mmWave Surface Wave Enabled Fluid Antennas and Experimental Results for Fluid Antenna Multiple Access}, 
      author={Yuanjun Shen and others},
      year={2024},
      journal = {arXiv e-prints, arXiv:2405.09663 [eess]},
}

@misc{PUMA26,
  author       = {\textcolor{blue}{Hanjiang Hong and Kai-Kit Wong and Xusheng Zhu and C. Rao and Dazhi He and Hyundong Shin}},
  title        = {{\textcolor{blue}{Phased ultra massive array (PUMA)}}},
  howpublished = {\textcolor{blue}{arXiv:2605.04866}},
  year         = {\textcolor{blue}{2026}},
}

@article{Dinis26,
  author  = {Daniel Dinis and Risto Wichman},
  title   = {{s-FAMA-GP: A Low-Complexity Slow FAMA Using Interference Interpolation}},
  journal = {IEEE Wireless Commun. Lett.},
    year={2026},
  volume={15},
  number={},
  pages={1727-1731},  doi     = {10.1109/LWC.2026.3661418}
}

@article{Liu2025,
    author = {\textcolor{blue}{Baiyang Liu and Kin-Fai Tong and Kai-Kit Wong and Chan-Byoung Chae and Hang Wong}},
    journal = {\textcolor{blue}{Opt. Express}},
    number = {\textcolor{blue}{13}},
    pages = {\textcolor{blue}{28898--28915}},
    publisher = {\textcolor{blue}{Optica Publishing Group}},
    title = {{\textcolor{blue}{Programmable meta-fluid antenna for spatial multiplexing in fast fluctuating radio channels}}},
    volume = {\textcolor{blue}{33}},
    month = {\textcolor{blue}{Jun}},
    year = {\textcolor{blue}{2025}},
    doi = {\textcolor{blue}{10.1364/OE.567923}}
}

@article{Heath2026TriHybrid,
  author  = {\textcolor{blue}{Heath, Robert W. and Carlson, Joseph and Deshpande, Nitish Vikas and Castellanos, Miguel Rodrigo and Akrout, Mohamed and Chae, Chan-Byoung}},
  title   = {{\textcolor{blue}{The Tri-Hybrid MIMO Architecture}}},
  journal = {\textcolor{blue}{IEEE Wireless Commun.}},
  year    = {\textcolor{blue}{2026}},
  volume  = {\textcolor{blue}{33}},
  number  = {\textcolor{blue}{1}},
  pages   = {\textcolor{blue}{199--206}},
  doi     = {\textcolor{blue}{10.1109/MWC.2025.3627138}}
}

@article{Castellanos2026TriHybrid,
  author  = {\textcolor{blue}{Castellanos, Miguel Rodrigo and Yang, Siyun and Chae, Chan-Byoung and Heath, Robert W.}},
  title   = {{\textcolor{blue}{Embracing Reconfigurable Antennas in the Tri-Hybrid MIMO Architecture for 6G and Beyond}}},
  journal = {\textcolor{blue}{IEEE Trans. Commun.}},
  year    = {\textcolor{blue}{2026}},
  volume  = {\textcolor{blue}{74}},
  pages   = {\textcolor{blue}{381--401}},
  doi     = {\textcolor{blue}{10.1109/TCOMM.2025.3621272}}
}

@article{New2025ChannelEstimation,
  author  = {New, Wee Kiat and Wong, Kai-Kit and Xu, Hao and Ghadi, Farshad Rostami and Murch, Ross D. and Chae, Chan-Byoung},
  title   = {Channel Estimation and Reconstruction in Fluid Antenna System: Oversampling Is Essential},
  journal = {IEEE Trans. Wireless Commun.},
  volume  = {24},
  number  = {1},
  pages   = {309--322},
  month   = jan,
  year    = {2025}
}
\end{document}